\documentclass[a4paper, 11pt,fleqn]{scrartcl}
\usepackage[left=2.5cm, right=2.5cm, top=2.5cm, bottom=2.5cm]{geometry}
\usepackage[english]{babel}
\usepackage[intlimits]{amsmath} 
\usepackage{amssymb} 
\usepackage{amsfonts} 
\usepackage{amsthm} 
\usepackage{fixltx2e} 
\usepackage{cite}
\usepackage{color}
\usepackage{graphicx}
\usepackage{enumerate}
\usepackage{hyperref}
\usepackage{natbib}
\usepackage{caption}
\usepackage{multirow}
\usepackage{longtable}
\usepackage{changes}
\usepackage{subfig}
\setcitestyle{authoryear} 

\newtheorem{definition}{Definition}[section]
\newtheorem{proposition}[definition]{Proposition}
\newtheorem{lemma}[definition]{Lemma}

\newtheorem{assumption}{Assumption}[section]
\newtheorem{regc}{Regularity Condition}[section]
\newtheorem{theorem}[definition]{Theorem}
\newtheorem{remark}{Remark}[section]
\newcommand{\var}{\mathrm{Var}}

\newcommand{\E}{E}

\DeclareMathOperator*{\argmax}{arg\,max}
\newcommand{\bs}{\boldsymbol}

\newcommand{\bt}{\boldsymbol{\theta}}
\newcommand{\tbt}{\boldsymbol{\widetilde{\theta}}}
\newcommand{\hbt}{\boldsymbol{\widehat{\theta}}}
\newcommand{\bb}{\boldsymbol{\beta}}

\newcommand{\hbb}{\widehat{\boldsymbol{\beta}}}
\newcommand{\M}{\boldsymbol{M}}
\DeclareOldFontCommand{\bf}{\normalfont\bfseries}{\mathbf}

\DeclareMathOperator{\sgn}{sgn}
\allowdisplaybreaks

\makeatletter

\def\blfootnote{\xdef\@thefnmark{}\@footnotetext}

\makeatother

\title{\textbf{Data Segmentation for Time Series Based on a General Moving Sum Approach}}
\author{Claudia Kirch$^{1,2}$ and Kerstin Reckruehm$^1$}
\date{\today}

\begin{document}

\maketitle

\begin{abstract}
	In this paper we propose new methodology for the data segmentation, also known as multiple change point problem, in a general framework including classic mean change scenarios, changes in linear regression but also changes in the time series structure such as in the parameters of Poisson-autoregressive time series. In particular, we derive a general theory based on estimating equations proving consistency for the number of change points as well as rates of convergence for the estimators of the locations of the change points. More precisely, two different types of MOSUM (moving sum) statistics are considered: A MOSUM-Wald statistic based on differences of local estimators and a MOSUM-score statistic based on a global estimator. The latter is usually computationally less involved in particular in non-linear problems where no closed form of the estimator is known such that numerical methods are required.  Finally, we evaluate the methodology by means of simulated data as well as using some geophysical well-log data.
\end{abstract}

\footnotetext[1]{Department of Mathematics, Otto-von-Guericke University;  Magdeburg, Germany. Supported by Deutsche Forschungsgemeinschaft - 314838170, GRK 2297 MathCoRe.}
\footnotetext[2]{E-mail: \url{claudia.kirch@ovgu.de}.} 

\section{Introduction}
Data segmentation or multiple change point estimation is a current topic in statistics and machine learning and comprises problems in a wide range of fields like finance, quality control, medicine or climate. For example, \cite{braunmueller} apply a technique of multiple change point detection on DNA sequences. \cite{aggarwal1999volatility} focus on finding structural breaks in the volatility of stock market returns, while \cite{killick2010detection} and \cite{killick2012optimal} give an interesting  application to oceanography by detecting changes in the variance of time series for wave heights.

Early literature on change point methodology focused on at-most-one-change testing and the corresponding estimator for a single change point (see e.g.\ \cite{CsorgoHorvath97}). Moving on to more complex data structures this is still an active field of research, see e.g. \cite{AueHorvath}, \cite{horvath2014extensions}, \cite{cho2020data}. At the same time, the theoretic investigation of the multiple change situation with all its additional challenges has become increasingly popular, where most papers deal with the detection of multiple changes in the mean; see for example the recent survey articles \cite{fearnhead2020relating}, \cite{cho2020data}.

One approach to the multiple change point in the mean problem  with favorable properties is based on moving sum (MOSUM) statistics first considered in detail by \cite{KirchMuhsal}. Corresponding tests had previously been considered by 
\cite{bauerhackl}, \cite{huvskova1990}, \cite{chuhornikkuan} and \cite{HuskovaSlaby01}.
\cite{ChoKirch2018} refine the procedure by an additional post-processing step based on an information criterion. 
An efficient implementation of the procedure in an \textsf{R}-package is detailed in \cite{ChoKirchMeier}. 
\cite{yau2016inference} use moving sum statistics in combination with an information criterion for estimating the change points in a linear autoregressive setting. In the context of renewal processes moving sum methodology was proposed by \cite{messer2014multiple} as well as \cite{KirchKlein}.

In this paper, we consider a general framework based on estimating equations similar to what has been considered in the at-most-one-change situation by \cite{kamgaing2016detection} as well as in a sequential change point context by \cite{kirch2015use,kirch2018modified}. This widely defined scope contains linear and non-linear (auto-)regressive time series including such based on neural network approximations. We propose two types of MOSUM (\underline{mo}ving \underline{sum}) statistics and derive  corresponding consistency results for estimators obtained from one bandwidth. Based on these results post-processing methods similar to \cite{ChoKirch2018} can be considered in future work.

The paper is organized as follows:
In Section~\ref{section_segmentation} we explain in detail the two data segmentation algorithms that we propose. The discussion in this paper occurs in a general framework based on estimating functions where examples are given in Subsection~\ref{subsection_examples}. In Section~\ref{section_consistency} we prove consistency of the data segmentation algorithms including localization rates for one of the procedures. Due to the general framework the consistency results are obtained based on some high-level assumptions. In Section~\ref{sectionasunderregcond} we prove these high-level assumptions for sufficiently smooth estimating functions under standard moment conditions. In Section~\ref{section_simulation} the empirical performance is investigated by means of a simulation study and using a geophysical well-data set before we give some conclusions in Section~\ref{section_conclusion}. The proofs can be found in an appendix.

\section{Data Segmentation Based on Moving Sum Statistics}\label{section_segmentation}
\subsection{MOSUM-Wald and MOSUM-Score Statistics}
We consider the following general setting 
\begin{align}\label{definitiongenparachamod}X_t=
	\sum_{j=1}^{q+1}X_t^{(j)}\, 1_{\{k_{j-1,n}<t\le k_{j,n}\}},\quad k_{0,n}=0,\; k_{q+1,n}=n.
\end{align}
	The sequences $\{X_t^{(j)}: {t\geq 1}\}, j=1,\ldots,q+1$, are assumed to be stationary, where consecutive sequences are distributionally different. Then,
$q$ denotes the number of structural breaks and $k_{1,n},\ldots,k_{q,n}$ denote the change points. 
For simplicity we assume that $q$ is fixed but all arguments are valid for a sequence $q_n$ as long as $q_n$ is bounded. Under some additional  assumptions it is also possible to relax the boundedness assumption on the number of change points (see e.g. \cite{KirchKlein} or Section E.3 in \cite{ChoKirch2018}).

While the construction of the below statistics is based on a given parametric model, we do not require the observed time series to follow that model. Instead we explicitly allow for model misspecification in the theoretic analysis showing that under misspecification changes are detected if consecutive time series have different best approximating model parameters (in the sense of Assumption~\ref{aswaldaltnullseg}).

The test statistic is related to estimators of the model parameters based on estimating functions $\bs{H}$ also known as M-estimators. These are obtained as the solution $\hbt_{a,b}$ of the estimating equation system $\sum_{i=a}^b\bs{H}(\mathbb{X}_i,\bt)\overset{!}{=}\bs{0}$ for a suitable choice of $\mathbb{X}_i$, which can be equal to the observations, a tuple of response and explanatory variables or (for autoregressive models) can include lagged observations; see Section~\ref{subsection_examples} for some examples.
The estimating function $\bs{H}$ is vector-valued for the estimation of multidimensional parameter vectors.

As we adopt a  general framework in this paper and explicitly allow for misspecification, we need to make some high-level assumptions on the underlying time series. We show their validity  under some moment conditions for sufficiently smooth estimating functions in Section~\ref{sectionasunderregcond}.
\begin{assumption}\label{as.invariancemulti}
	\begin{enumerate}[(a)]
		\item	For each segment $j$, let $\{\mathbb{X}_t^{(j)}:t\geq 1\}$ be stationary.
		\item 
			For a given $\boldsymbol{\theta}$ specified below  let \(\boldsymbol{S}_j(k,\boldsymbol{{\theta}})=\sum_{i=1}^{k}\boldsymbol{H}(\mathbb{X}_i^{(j)}, \boldsymbol{{\theta}})\) fulfill a strong invariance principle for all $j=1,\ldots,q+1$, i.e.\ possibly after changing the probability space there exists some  \(\nu>0\), a \(p\)-dimensional standard Wiener process \(\{\boldsymbol{W}(k): k \geq 0\}\) and a symmetric positive definite long-run covariance matrix $\boldsymbol{\Sigma}_{(j)}(\bt)$ such that as $k\to\infty$
\begin{align*}
\left\lVert\boldsymbol{S}_j(k,\boldsymbol{{\theta}})-\E(\boldsymbol{S}_j(k,\boldsymbol{{\theta}}))- \boldsymbol{\Sigma}_{(j)}(\bt)^{1/2}\, \boldsymbol{W}(k)\right\rVert=O(k^{1/(2+\nu)})\;\; a.s.\end{align*}
	\end{enumerate}
\end{assumption}
In the literature invariance principles as in (b) have been derived for many time series --  compare also Theorem~\ref{theorem_ass_ip} below, where typically the parameter $\nu$ depends on the number of existing moments. 

 The following assumption is a Bahadur representation for M-estimators as for example derived by \cite{he1996general} but uniformly in $k$. Effectively, this is a linearisation of the estimator which is typically used to derive asymptotic normality. In our case, the representation will also be used to make a connection between the MOSUM-Wald and the MOSUM-score statistics in the proof.

\begin{assumption}\label{aswaldaltnullseg}
	There exists a regular matrix $\boldsymbol{V}_{(j)}$ for any $j=1,\ldots, q+1$ such that
\begin{align*}
	&\max_{k_{j-1,n}<k\le k_{j,n}-G}\left\lVert\sqrt{\frac G 2}\boldsymbol{V}_{(j)}\left( \boldsymbol{\theta}_j-\boldsymbol{\widehat{\theta}}_{k+1,k+G}\right)-\frac{1}{\sqrt{2G}}\sum_{i=k+1}^{k+G}\boldsymbol{H}(\mathbb{X}^{(j)}_i,\boldsymbol{\theta}_j)\right\rVert\\&=  o_P\left((\log(n/G))^{-1/2}\right),
\end{align*}
where $\bt_j$ fulfills $\E\left(\bs{H}(\mathbb{X}_i^{(j)},\bt_j)\right)=\bs{0}$ (identifiably unique). \end{assumption} 

If the parametric assumption underlying the M-estimator is correct, then $\bt_j$ is the true underlying parameter for the $j$th stationary sequence. Otherwise, this is the best-approximating parameter in the sense induced by the estimating function. Furthermore, in case of differentiable estimating functions, it usually holds $\boldsymbol{V}_{(j)}=\E\left( \nabla \boldsymbol{H} (\mathbb{X}^{(j)}_{1},\boldsymbol{\theta}_j)^T\right)$; see Regularity Condition~\ref{regc_wald}~(b).

Moreover, we define
\begin{align}
	&	\boldsymbol{V}_k=\boldsymbol{V}_{(j)}, \; \boldsymbol{\Sigma}_k(\bt)=\boldsymbol{\Sigma}_{(j)}(\bt),\; \boldsymbol{\Sigma}_k=\boldsymbol{\Sigma}_{(j)}(\bt_j)\quad \text{for }k_{j-1,n}<k\le k_{j,n},\notag\\
	&	\boldsymbol{\Gamma}_k=\boldsymbol{V}_{k}^{-1}\boldsymbol{\Sigma}_{k}\left(\boldsymbol{V}_{k}^{-1}\right)^T, \qquad \boldsymbol{\Gamma}_{(j)}=\boldsymbol{\Gamma}_{k_{j,n}},\label{eq_Gamma}
\end{align}
where $\boldsymbol{\Gamma}_k$ is the asymptotic (long-run) covariance matrix for the estimators from the $j$th stationary segment  under the above assumptions.

A first intuitive  approach to the multiple change problem is based on the following
\textbf{MOSUM-Wald statistic} that uses weighted differences of moving estimators $\boldsymbol{\widehat{\theta}}_{k+1,k+G}$  and $\boldsymbol{\widehat{\theta}}_{k-G+1,k}$ of the unknown quantity $\theta$ based on the stretch of data $X_{k+1},\ldots,X_{k+G}$ respectively $X_{k-G+1},\ldots,X_k$. Similarly to the Mahalanobis distance the weighting is done with the asymptotic (long-run) covariance matrix, so that  we obtain for $k=G,\ldots,n-G$
\begin{align}\label{DefWaldstatistic}
	T^{(1)}_{k,n}(G)=T^{(1)}_{k,n}(G;\boldsymbol{\Gamma}_k)=\frac{\sqrt{G}}{\sqrt{2}}\left\lVert\boldsymbol{\Gamma}_{k}^{-1/2}\left(\boldsymbol{\widehat{\theta}}_{k+1,k+G}-\boldsymbol{\widehat{\theta}}_{k-G+1,k}\right)\right\rVert,\end{align}
where the bandwidth $G=G_n$ is a tuning parameter that determines the length of the moving window and $\|\cdot\|$ denotes the Euclidian norm.

In this paper, we consider bandwidths $G$ that are of smaller order than the sample length $n$. 
The procedure detects changes that are further apart than $2G$ such that sublinear changes, whose distance diverges strictly slower than the sample size, can  be detected. 
\begin{assumption}\label{as.bandwidth}
	\begin{enumerate}[(a)]
		\item 	For $\nu>0$ (specified in Assumption~\ref{as.invariancemulti} (b)) let 
	\begin{align*}
\frac{n}{G}\rightarrow \infty\;\;\text{and}\;\; \frac{n^{\frac{2}{2+\nu}}\log(n)}{G}\rightarrow 0\;\;\text{for}\;\;n\rightarrow
\infty.
\end{align*}
\item The minimal distance between two neighbouring change points is asymptotically larger than $2G$ in the sense of
	$$\lim\inf_{n\to\infty}\min_{j=1,\ldots,q+1}(k_{j,n}-k_{j-1,n})/G>2.$$
	\end{enumerate}
\end{assumption}

Close to a change point  the above Wald statistic \eqref{DefWaldstatistic} can be expected to be large such that it can be used to find changes in the parameter vector $\bt$. However, using the Wald statistic has one major drawback: Calculating  two estimates for each time point, i.e. $2(n-2G)$ in total, can be computationally challenging in situations where numerical methods need to be applied as e.g.\  in non-linear models such as Poisson autoregressive models. Furthermore, in situations with many (almost) zeroes (corresponding to many local optima) such as e.g.\ in a neural-network-autoregressive situation as in \cite{kirch2012testing} estimators can be far apart even though the corresponding stochastic processes are almost identical. While this is explicitly excluded for our theoretic results by assuming identifiability, this can easily lead to problems   in applications.

In order to avoid these problems and to reduce the computational complexity and corresponding numerical challenges we  consider \textbf{MOSUM-score statistics} where local parameter estimators are replaced by a global inspection parameter $\boldsymbol{\widetilde{\theta}}$, which can  be fixed or an estimator based on the same data. The MOSUM-score statistic is based on the following differences between moving sums of the estimating function at the inspection parameter 
for $k=G,\ldots,n-G$
\begin{align*}
	\boldsymbol{M}_{\boldsymbol{\widetilde{\theta}}}(k)=\sum_{i=k+1}^{k+G}\boldsymbol{H}(\mathbb{X}_i,\boldsymbol{\widetilde{\theta}})- \sum_{i=k-G+1}^{k}\boldsymbol{H}(\mathbb{X}_i,\boldsymbol{\widetilde{\theta}}).
\end{align*}
Effectively, this converts a general multiple parameter change problem to a multiple mean change problem of the transformed sequence $\{\bs{H}(\mathbb{X}_t,\tbt)\}_{t\geq 1}$. Consequently, a multivariate version of the (univariate) mean-MOSUM statistic investigated by \cite{KirchMuhsal} can be applied, leading to the following MOSUM-score statistic:
\begin{align}\label{Defscorestatistic}T^{(2)}_{k,n}(G)=T^{(2)}_{k,n}(G,\boldsymbol{\widetilde{\theta}})=T^{(2)}_{k,n}(G,\boldsymbol{\widetilde{\theta}};\boldsymbol{\Sigma}_k(\tbt))=\frac{1}{\sqrt{2G}}\left\lVert\boldsymbol{\Sigma}_k(\tbt)^{-1/2}\boldsymbol{M}_{\boldsymbol{\widetilde{\theta}}}(k)\right\rVert.\end{align}

If we use data-dependent inspection parameters ${\tbt}_n={\tbt}_n(\mathbb X_1,\ldots,\mathbb X_n)$, then we need the following additional assumptions for the MOSUM-score procedure:
\begin{assumption}\label{as.replacemulti}
There exists $\tbt$ such that for any $j=1,\ldots,q$
\begin{align*}
	&(i)\quad\max_{k_{j-1,n}+G\le k\le k_{j,n}-G}\left\lVert \M_{\tbt_n}(k)-\M_{\tbt}(k)\right\rVert=o_P\left(\sqrt{\frac{G}{\log(n/G)}}\right),\\
	&(ii)\quad \max_{|k-k_{j,n}|< G}\left\lVert \M_{\tbt_n}(k)-\M_{\tbt}(k)\right\rVert=o_P\left(\sqrt{G\,\log(n/G)}\right).
\end{align*}
\end{assumption}

This assumption holds under smoothness and mixing assumptions for any $\sqrt{n}$-consistent estimator $\tbt_n$ (see  Theorem~\ref{theorem_ass_score} below). In particular,  $\tbt_n$  obtained from the corresponding estimating equations based on the stretch of  data $X_a,\ldots,X_b$ are $\sqrt{n}$-consistent for  $\tbt$, 
the best-approximating parameter in the sense of Theorem~\ref{theoremrootnconsisalt}. While for the latter approach some theoretical guarantees towards detectability are given in Section~\ref{sectiondetectabibilty}, from a computational perspective it might be more efficient to use different estimators (for an empirical example, see Section~\ref{section_det_sim}).

The latter MOSUM-score approach avoids the numerical drawbacks of the MOSUM-Wald statistic but can only detect changes for which the given  inspection parameter $\tbt$ results in a change in the expectation of the transformed series. This problem will be discussed in detail in Section \ref{sectiondetectabibilty}.

In the mean change model based on moving sample means as considered by \cite{KirchMuhsal} (see also Example~\ref{examplemeanchange} below) the above MOSUM-Wald and MOSUM-score statistics coincide for any inspection parameter $\tbt$.

\subsection{Segmentation Algorithm}\label{sec_seg_alg}
MOSUM statistics as introduced in the previous section have peaks close to the true change points making them particularly suitable for data segmentation. More precisely, as demonstrated in Figure~\ref{figure_mosum}, the MOSUM statistic is a noisy version of the MOSUM signal, which is a piecewise linear function that is equal to zero away from the change points and has a single peak at each change point.

\begin{figure}
\includegraphics[width=\textwidth,trim=0 1.7cm 0 1.7cm]{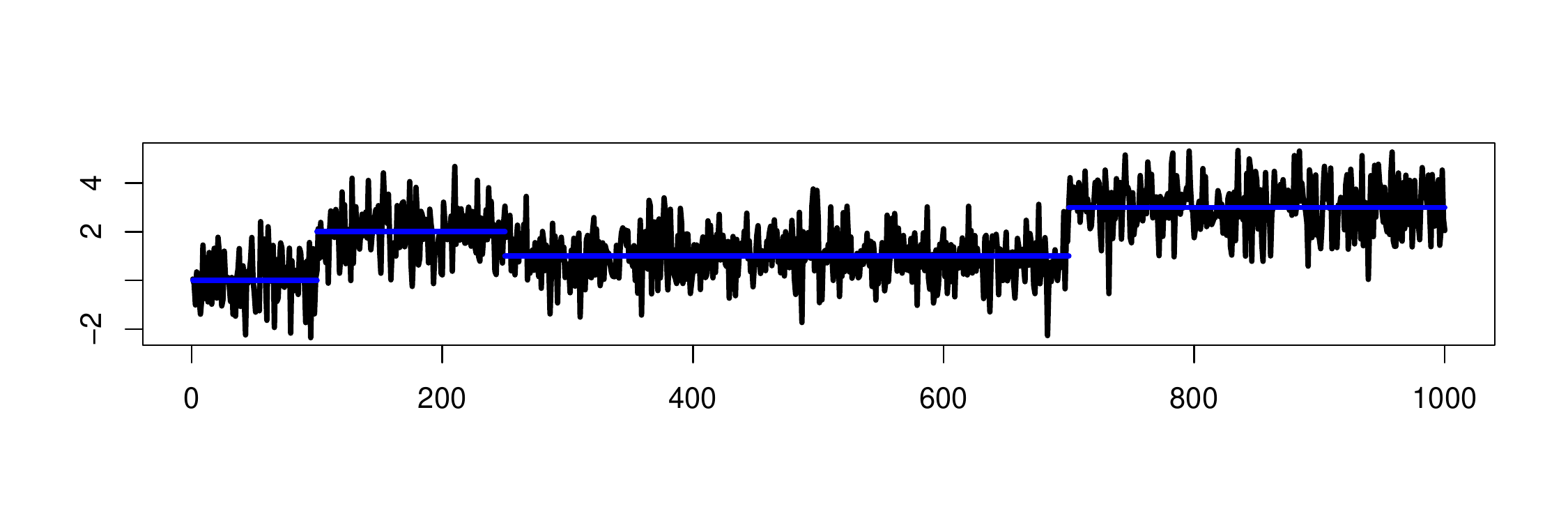} 
\includegraphics[width=\textwidth,trim=0 1.7cm 0 1.7cm]{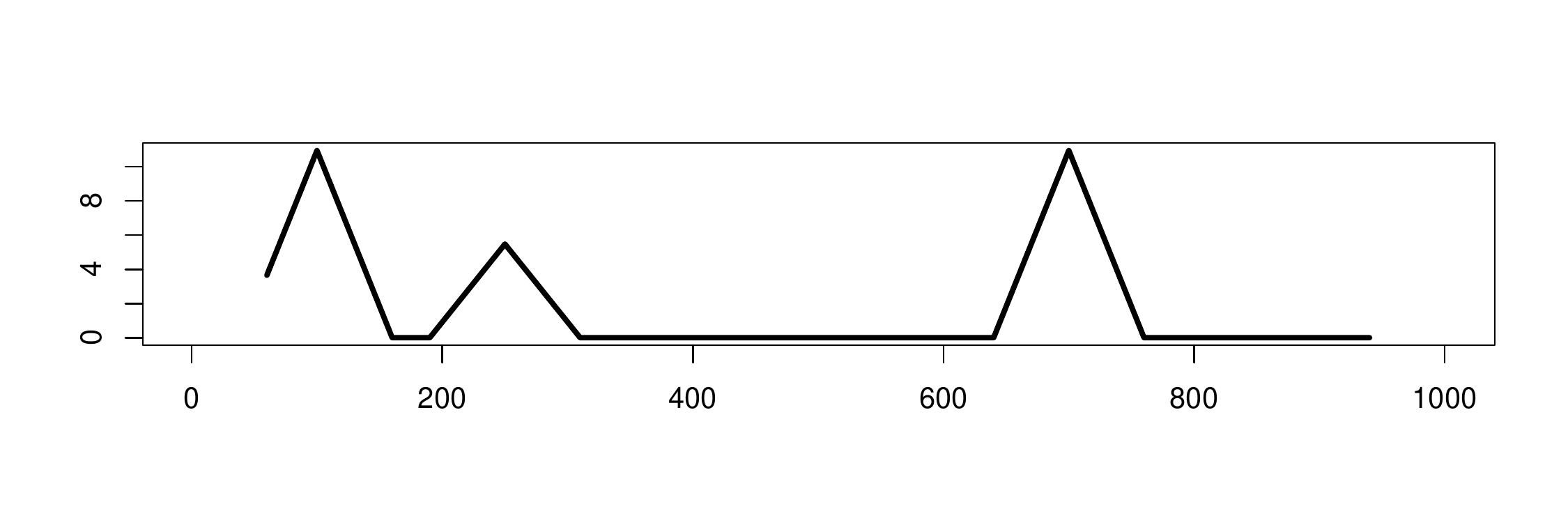} 
\includegraphics[width=\textwidth,trim=0 1.7cm 0 1.7cm]{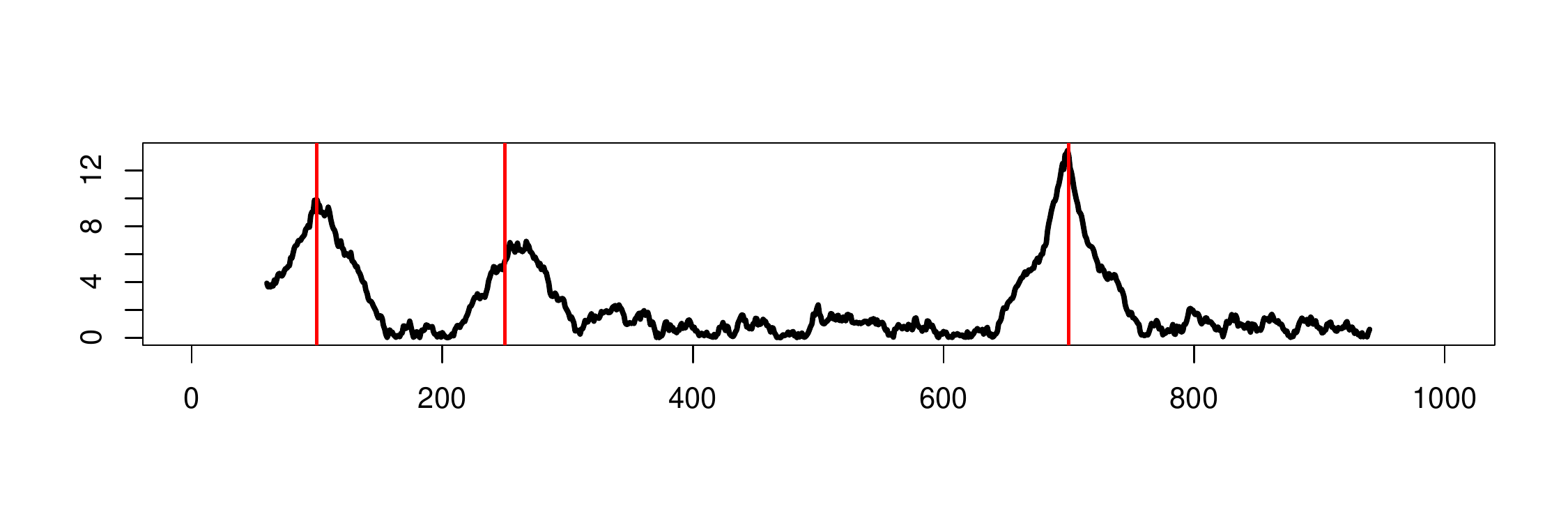}
\includegraphics[width=\textwidth,trim=0 1.7cm 0 1.7cm]{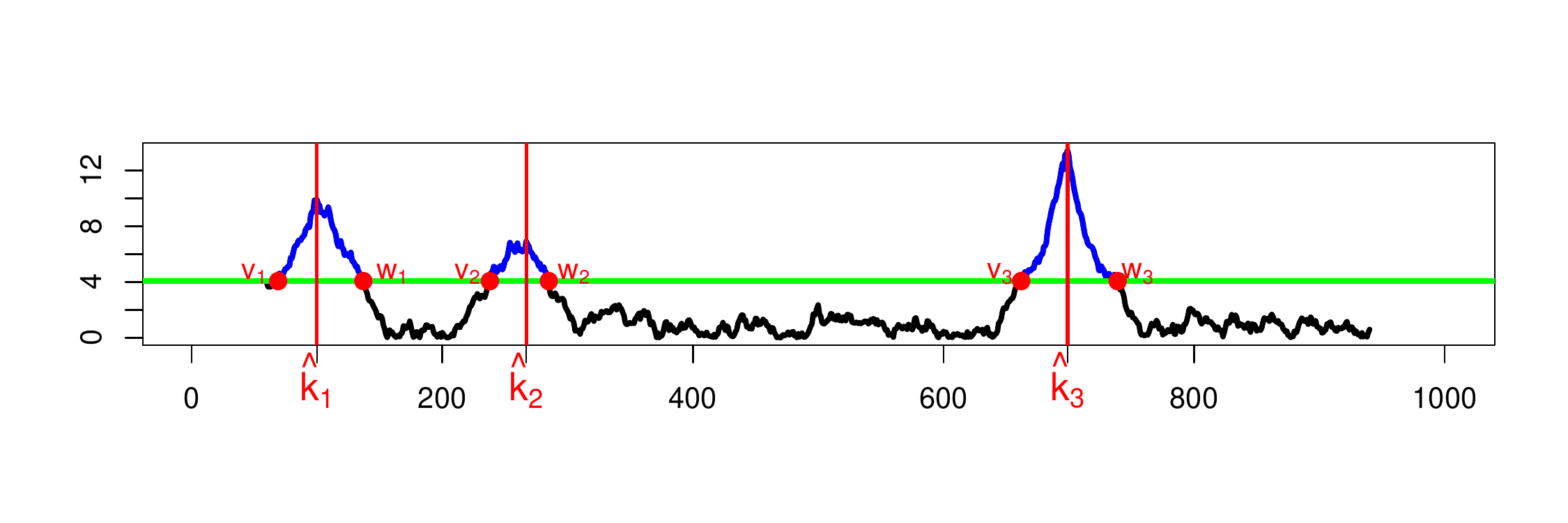}
\caption{Top panel: Time series with changes in the mean at time points $100$, $250$ and $700$. Second panel: Signal of the classical MOSUM statistic for mean changes as investigated by \cite{KirchMuhsal} with bandwidth $G=60$. Third panel: Actual (noisy) MOSUM statistic for the above data set (the vertical lines indicate the true change points). Fourth panel: Outcome of the segmentation procedure as described in Section~\ref{sec_seg_alg} in this example.}
	\label{figure_mosum}
\end{figure}

Clearly, in order to use this noisy signal for segmentation purposes a threshold is needed  to distinguish between significant local maxima that are due to a close-by change point and local maxima obtained simply from random fluctuations around zero if no change point is close by. We  obtain such thresholds by globally controlling the random fluctuations of the MOSUM statistic asymptotically if no change point is present in the time series. As a consequence the proposed procedure controls the (asymptotic) family-wise error rate at level $\alpha_n$ (see below) for the detected change points.
In Section~\ref{sectionthreshold} below we detail how to choose such a threshold $D_n(\alpha_n,G)$ based on asymptotic  $\alpha_n$-quantiles (of the no-change situation). 
These quantiles can also be considered as critical values for a corresponding uniform (across time) test procedure for the null hypothesis of no change versus the alternative of one or more changes. An asymptotic threshold is reasonable given the generality of our procedure, the (non-parametric) error structure and as it also allows for model misspecification.

Based on a MOSUM-Wald or MOSUM-score  statistic $T_{k,n}(G)$ with corresponding threshold $D_n(\alpha_n,G)$ we propose the following \textbf{MOSUM segmentation procedure} to determine estimators for the number and the locations of the change points. Below, we suppress the possible dependence on other tuning parameters such as the inspection parameter $\boldsymbol{\widetilde{\theta}}$ for the MOSUM-score statistic:\vspace{2mm}

\textit{Consider all pairs of time points $(v_{j,n},w_{j,n})$, $v_{1,n}< w_{1,n}<v_{2,n}<w_{2,n}<\ldots$, with 
  \begin{align}& \;T_{k,n}(G)\geq D_n(\alpha_n,G)\;\; \text{for}\;\; v_{j,n}\leq k \leq w_{j,n},\notag\\
   &\; T_{k,n}(G)< D_n(\alpha_n,G)\;\; \text{for}\;\;k=v_{j,n}-1,w_{j,n}+1,\notag\\
  &\;  w_{j,n}-v_{j,n} \geq \varepsilon G \;\;\;\text{with}\;\; 0<\varepsilon<1/2 \;\;\text{arbitrary but fixed}.\label{MOSUMcond3}
   \end{align}
   We take the number of these pairs as an estimator for the number of changes:\\
   $$\widehat{q}_n=\;\;\text{number of pairs}\;(v_{j,n},w_{j,n}).$$
   Furthermore, we determine the local maxima between $v_{j,n}$ and $w_{j,n}$, $j=1,\ldots,\widehat{q}_n$, and use them as estimators for the locations of the change points:
   $$\widehat{k}_{j,n}=\argmax\limits_{v_{j,n}\leq k \leq w_{j,n}} T_{k,n}(G)$$}

The pairs $(v_{j,n},w_{j,n})$, $j=1,\ldots,\widehat{q}_n$, give start and end points of intervals on which the statistic exceeds the threshold. For this reason we call $[v_{j,n},w_{j,n}]$, $j=1,\ldots,\widehat{q}_n$, intervals of exceedings or exceeding intervals.

Condition \eqref{MOSUMcond3} avoids false positives obtained due to random fluctuation when the linear signal crosses the threshold. While the linear signal is strictly monotone when crossing the threshold, the noisy version may not be -- possibly jumping beneath the threshold and then above again simply by chance.  

By the below theory (see Section~\ref{section_consistency}), we can be relatively certain (with asymptotic probability one), that the exceeding intervals contain a change point.  Typically $T_{k,n}(G)$ depends on an estimator for the long-run covariance matrix $\bs{\Gamma}_k$ resp.\ $\boldsymbol{\Sigma}_k$, where in many situations a local estimator (depending on the time point $k$) is required because the long-run covariance matrix depends on the segment. However, once the exceeding intervals have been determined, different and possibly better estimators  $\widehat{\boldsymbol{\Psi}}_{j,n}$, $j=1,\ldots,\widehat{q}_n$, can be used to obtain the final change point estimator
\begin{align}\label{eq_est_Psi}
	\widehat{k}_{j,n}(\widehat{\boldsymbol{\Psi}}_{j,n})=\argmax\limits_{v_{j,n}\leq k \leq w_{j,n}} T_{k,n}(G;\widehat{\boldsymbol{\Psi}}_{j,n}),
\end{align}
where $\widehat{\boldsymbol{\Psi}}_{j,n}$ is a sequence of symmetric positive definite matrices fulfilling Assumptions~\ref{ass_Psi} below.
For a real-valued estimating function $\mathbf{H}$ this is always equivalent to using the unscaled maximizer, i.e.\ $\widehat{\boldsymbol{\Psi}}_{j,n}=1$. For a vector-valued estimating function it is also possible (e.g. for computational reasons) to use the identity matrix $\widehat{\boldsymbol{\Psi}}_{j,n}=\mathbf{\mbox{Id}}$ without jeopardizing the asymptotic localization rates.
Indeed, the  theory does not even require the $\widehat{\boldsymbol{\Psi}}_{j,n}$ to be consistent estimators for the long-run covariance matrix.
 While we use the initial estimator in the simulation study, we obtain the localization rates for the latter.  A corresponding analysis for covariance estimators  depending on the time point $k$ is not possible in our general framework, but can be done for specific estimators (see e.g. Remark 3.2 and Corollary 3.1 in \cite{KirchMuhsal} for an example).

 \subsection{Examples}\label{subsection_examples}
The framework that is used in this work covers many different situations (see e.g. \cite{kirch2015use} and \cite{kirch2018modified} for additional examples). Here, we concentrate on the following three examples: Uni- and multivariate  changes in expectation where the sample mean but also more robust estimators can be used, linear regression and Poisson autoregression.
 
\subsubsection{Mean Change Model}\label{examplemeanchange}
	
The univariate mean change model is given by 
\begin{align*}
	X_i=\sum_{j=1}^{q+1}\mu_j\, 1_{\{k_{j-1,n}<i\le k_{j,n}\}} + \varepsilon_i,
\end{align*}
		where $\{\varepsilon_i: i\geq 1\}$ is a centered stationary error sequence with autocovariance function $\gamma(\cdot)$ and long-run variance $0<\tau^2=\sum_{h\in \mathbb{Z}}\gamma(h)<\infty$. Thus, the expected value of each stretch is given by  $\mu_j\neq \mu_{j+1}$, $j=1,\ldots,q$. 
		The MOSUM statistic discussed in \cite{KirchMuhsal} is based on the sample mean $\bar{X}_{1,n}=\frac{1}{n}\sum_{i=1}^nX_i$ which is the solution of the estimating equation $\sum_{i=1}^n(X_i-\mu)\overset{!}{=}0$, hence $\mathbb{X}_i=X_i$.

 An alternative M-estimator for the expectation for symmetric error distributions is based on the  estimating function $H(X_i,\mu)=\frac{2}{\pi}\arctan(\mu-X_i)$. This estimating function leads to more robust estimators and can be seen as a smooth approximation of the sign-function that leads to the median.
 Henceforth, we  will therefore call the corresponding estimator \emph{median-like estimator}.

 A multivariate version $\boldsymbol{X}_i=(X_{1,i},\ldots,X_{d,i})^T$ is obtained by replacing $\varepsilon_j$ as well as $\mu_j$ with multivariate quantities.  The multivariate estimating function  is given by\\ $\boldsymbol{H}(\boldsymbol{X}_i,\boldsymbol{\mu})=(H(X_{1,i},\mu_1),\ldots,H(X_{d,i},\mu_d))^T$ for $\boldsymbol{\mu}=(\mu_1,\ldots,\mu_d)^T$, i.e.\  $\mathbb{X}_i=\boldsymbol{X}_i$. Effectively, the procedure uses the corresponding univariate estimators in each component.  Consequently, an estimator for the long-run covariance matrix is now required which poses a major difficulty in practice.

\subsubsection{Linear Regression Model with Structural Breaks}\label{examplelinearregression}
The  linear regression model with changes is given by 
\begin{align*}
	X_i=\sum_{j=1}^{q+1}\bs{Z}_i^T\bb_j\, 1_{\{k_{j-1,n}<i\le k_{j,n}\}} + \varepsilon_i,
\end{align*}
		where $X_i$ denotes the response variable, $\bs{Z}_i$ are the regressors, $\bb_1,\ldots,\bb_{q+1}$ are the parameter vectors and $\{\varepsilon_i\}_{i\geq 1}$ represents an innovation sequence as in Section \ref{examplemeanchange}. 		
A least squares approach minimising the sum of squared residuals corresponds to the vector-valued estimating function 
\begin{align*}
	\bs{H}\left((X_i,\bs{Z}_i)^T,\bb\right)=-2\bs{Z}_{i}\left(X_i-\bs{Z}_{i}^T\bb\right),
\end{align*}
i.e.\ $\mathbb{X}_i=(X_i,\bs{Z}^T_{i})^T$. For the mathematical results some additional assumptions on the regressors such as stationarity and independence from the innovation sequence are required.

\subsubsection{Poisson Autoregressive Model with Structural Breaks}\label{examplepoiautoreg}
In this example we consider the Poisson autoregressive model of order one, which is a widespread model for integer-valued time series.

A Poisson autoregressive model of order one, also known as INARCH(1) model, with $q$ change points can be described by \eqref{definitiongenparachamod}, where 
$\{X_i^{(j)}\}$ are INARCH(1) time series with parameters $\bt_j=(\theta_{j,1},\theta_{j,2})^T$, $j=1,\ldots,q+1$, i.e.
\begin{align*}X_i^{(j)}|\mathcal{F}_{i-1} \sim Poi(\lambda_i),\quad\text{with } \lambda_i=\theta_{j,1}+\theta_{j,2}X_{i-1}.\end{align*}
	In this model the observation $X_i$ conditioned on the past is Poisson distributed with a parameter $\lambda_i$ depending on the past observation. A (partial) likelihood approach leads to the  estimating function  
	$$\bs{H}((X_i,X_{i-1})^T,\bt)=-2\bs{X}_{i-1} \left(\frac{X_i}{\bs{X}_{i-1}^T\bt} -1\right),$$
	where $\bs{X}_{i-1}=\left(1,X_{i-1}\right)^T$, i.e.\ here $\mathbb{X}_i=(X_i,X_{i-1})^T$.

\subsubsection{Further Examples}
In the context of at-most-one-change a posteriori  as well as sequential change point analysis, many procedures proposed in the literature are of this type (see \cite{KirchTadjHandb} as well as \cite{kirch2015use} for more examples), such that their extensions to the data segmentation problem is covered by this framework.

In the context of the mean change model as in Section~\ref{examplemeanchange}, in many applications the variance of the error sequence $\var(\varepsilon_i)=\sigma^2$ changes in addition to the mean. Thus, it is reasonable to assume that the variance is only piecewise constant such that a change point can be caused by a mean change, a variance change or a change in both parameters. In order to detect these types of changes one can use the following vector-valued estimating function 
$\bs{H}(X_i,\mu,\sigma^2)=\left(X_i-\mu,(X_i-\mu)^2-\sigma^2\right)^T$ (which relates to the method of moments but also the maximum likelihood estimator under Gaussianity assumption).  A related MOSUM-Wald statistic has recently been proposed and analyzed by \cite{messer2019bivariate}.

\subsection{Threshold Selection}\label{sectionthreshold}
In the following theorem we derive the limit distribution of the maximum of the MOSUM statistics if no structural break occurs. 
To this end, let
\begin{align}\label{defaandbfunction}
&a(x)= \sqrt{2\log(x)} \text{ and }\\
&b(x)=2\log(x)+\frac{p}{2}\log(\log(x))-\log\left(\frac{2}{3}\,\Gamma\left(\frac{p}{2}\right)\right),\notag
\end{align}
where $p$ is the dimension of the parameter space and $\Gamma$ denotes the gamma function.

\begin{theorem}\label{limitdistribution}
	Consider model~\eqref{definitiongenparachamod} with no change, i.e.\ $q=0$.
	Let Assumptions~\ref{as.bandwidth} on the bandwidth hold, in addition to
	Assumptions \ref{as.invariancemulti}, where (b) needs to hold with $\bt=\bt_1$ (as defined in Assumption~\ref{aswaldaltnullseg}) for the MOSUM-Wald and with $\bt=\tbt$ for the MOSUM-score statistic. For the MOSUM-Wald statistic, let also Assumption~\ref{aswaldaltnullseg} hold.
	\begin{enumerate}[(a)]
	 \item Then, for both the MOSUM-Wald ($\ell=1$) as well as the MOSUM-score ($\ell=2$ with $T_{k,n}^{(2)}(G)=T_{k,n}^{(2)}(G,\tbt)$) statistic it holds
	\begin{align*}
		a(n/G)\,\max_{k=G,\ldots,n-G}T^{(\ell)}_{k,n}(G)-b(n/G)\; \xrightarrow{\mathcal{D}}\; E
\end{align*}
with \(E\) a Gumbel distributed random variable fulfilling \(P(E\leq x)= \exp(-2\exp(-x))\) and with $a(x)$ and $b(x)$ as in \eqref{defaandbfunction}.
\item For the MOSUM-score statistic the result remains true for data-dependent inspection parameters $\tbt_n$ if Assumption~\ref{as.replacemulti} holds. Additionally, the assertions remain true if the long-run covariance matrix $\boldsymbol{\Sigma}_{1}$ is replaced by an estimator fulfilling (with $\|\cdot\|$ denoting the spectral norm of a matrix)	\begin{align*}
\max_{G\leq k \leq n-G}\left\lVert \widehat{\boldsymbol{\Sigma}}_{k,n}^{-1/2} - \boldsymbol{\Sigma}_1^{-1/2}\right\rVert=o_P\biggl(\bigl(\log(n/G)\bigr)^{-1}\biggr).
\end{align*}
\item The assertion for the MOSUM-Wald statistic remains true if both $\boldsymbol{\Sigma}_{1}$ and $\boldsymbol{V}_1$ are replaced by estimators $\widehat{\boldsymbol{\Sigma}}_{k,n}$ respectively $\widehat{\boldsymbol{V}}_{k,n}$ fulfilling
	\begin{align}\label{eq_Gamma_null}\max_{G\leq k \leq n-G}\left\lVert\boldsymbol{\Sigma}_1^{-1/2}\bs{V}(\bt_1)-\widehat{\boldsymbol{\Sigma}}_{k,n}^{-1/2}\widehat{\boldsymbol{V}}_{k,n}\right\rVert=o_P\left(\left(\log(n/G)\right)^{-1}\right).\end{align}
\end{enumerate}
\end{theorem}

As threshold in our segmentation algorithm as described in Section~\ref{sec_seg_alg}  we use the asymptotic $(1-\alpha_n)$-quantiles as given in Theorem~\ref{limitdistribution}  i.e.\ \begin{align}\label{eq_neu_threshold}
	D_n(\alpha_n,G)= \frac{b(n/G)+c_{\alpha_n}}{a(n/G)},\qquad c_{\alpha_n}=-\log\log\frac{1}{\sqrt{1-\alpha_n}},
\end{align}
where $c_{\alpha_n}$ is the $(1-\alpha_n)$-quantile of the limiting Gumbel distribution.
Furthermore, we require
$\alpha_n\to 0$ sufficiently slowly in the sense of: \begin{assumption}\label{as.siglevel}
Let the sequence of significance levels $\alpha_n$ fulfill
\begin{align*}
	\alpha_n \rightarrow 0\;\;\; \text{and} \;\;\;c_{\alpha_n}=O(b(n/G)).\end{align*}\end{assumption}

\begin{remark}
For fixed $\alpha$ the threshold $D_n(\alpha,G)$ can be used as a critical value in a test procedure based on MOSUM statistics for the null hypothesis of no change versus the alternative of at least one change.
Proposition~\ref{propositionsprobminmaxwald} below shows when such a test procedure has asymptotic power one.
 While the power is usually not as good as for the corresponding at-most-one-change statistics even in the presence of multiple changes, the latter are not as suitable for the localization of change points.
\end{remark}

\section{Consistency of the MOSUM Segmentation Procedures}\label{section_consistency}\subsection{MOSUM-Wald Procedure}\label{chapterwald}
Part (a) of the following assumption  controls the behavior of estimators that are not contaminated by change points and extends standards results from a pointwise to a uniform (within a $G$-environment) situation. It follows from Assumption~\ref{aswaldaltnullseg} under weak assumptions  on the time series and estimating functions. On the other hand, for contaminated estimators (containing observations from two segments), the estimator cannot be expected to converge to either of the best approximating parameters but will be contaminated by both (see Theorem~\ref{theoremrootnconsisalt} below). Part (b) of the below assumption guarantees that the estimator differs from the best approximating parameter of the $j$th segment if at least $\varepsilon G$ of the summands are from a neighboring regime.
 For sufficiently smooth estimating functions and mixing time series, the validity of the assumption will be shown in Section~\ref{sectionasunderregcond} below.
 \begin{assumption}\label{Ass_31_new}For any $j=1,\ldots,q$ let
	 \begin{align*}
&(a)\qquad
\max_{\substack{k_{j-1,n}<k\le k_{j-1,n}+G \\ k_{j,n}-2G\le k\le k_{j,n}-G}}\sqrt{G}\,\|\boldsymbol{\widehat{\theta}}_{k+1,k+G}- \boldsymbol{\theta}_j\|=O_P\left( \sqrt{\log(n/G)} \right).\\
&(b)\qquad 
\sqrt{\frac{G}{\log(n/G)}}\min_{\substack{k_{j-1,n}-G\le k\le k_{j-1,n}-\varepsilon\,G \\ k_{j,n}-(1-\varepsilon)\,G\le k\le k_{j,n}}}\|\boldsymbol{\widehat{\theta}}_{k+1,k+G}- \boldsymbol{\theta}_j\|\overset{P}{\longrightarrow}\infty.
\end{align*}\end{assumption}

Because the long-run covariance matrix $\bs{\Gamma}_k=\boldsymbol{V}_{k}^{-1}\boldsymbol{\Sigma}_{k}\left(\boldsymbol{V}_{k}^{-1}\right)^T$ is usually unknown in applications we need to understand the influence of estimating this quantity on the procedure. This is reflected by the following assumption stating that the estimator needs to be consistent away from changes (related to \eqref{eq_Gamma_null}), while we can allow for a certain amount of misspecification close to the changes.
\begin{assumption}\label{asaltestcovmatwald}
Let \(\widehat{\boldsymbol{\Gamma}}_{k,n}\) be an estimator sequence of $\bs{\Gamma}_k$ satisfying
 
\begin{enumerate}[(a)]
\item   $\max_{k_{j-1,n}+G\le k\le k_{j,n}-G}\left\lVert \boldsymbol{\widehat{\Gamma}}^{-1/2}_{k,n}-\boldsymbol{\Gamma}_k^{-1/2}\right\rVert=o_P\left(\log(n/G)^{-1}\right)$\\[1mm] for any $j=1,\ldots,q+1$,
\item for any $j=1,\ldots,q$ that
	\begin{align*}
		&\sup_{|k-k_{j,n}|\le (1-\varepsilon) G}\left\lVert\boldsymbol{\widehat\Gamma}_{k,n}^{1/2}\right\rVert<\infty,\qquad \sup_{|k-k_{j,n}|\le (1-\varepsilon) G}\left\lVert\boldsymbol{\widehat\Gamma}_{k,n}^{-1/2}\right\rVert<\infty.
\end{align*}
\end{enumerate}
\end{assumption}

\begin{remark}\label{rem_cov_wald}
	In Assumption~\ref{asaltestcovmatwald} (a) we require a uniformly consistent covariance estimator  at locations well away from change points. This assumption allows to work with the minimal threshold as discussed in Section~\ref{sectionthreshold}  of exact order $\sqrt{\log(n/G)}$ which in turns leads to minimal assumptions on the signal strength in Assumption~\ref{Ass_31_new}(b).
	In practice, it can be difficult or numerically expensive to use such estimators (see also Section~\ref{section_simulation}). 
In such cases, we can replace this assumption  by the weaker one that
	\begin{align*}
		\max_{k_{j-1,n}+G\le k\le k_{j,n}-G}\left\lVert \boldsymbol{\widehat{\Gamma}}^{-1/2}_{k,n}\right\rVert<\infty.
	\end{align*}
	Then, the below results remain true as long as we use a 
slightly larger threshold $\widetilde{D}_n(G)$ that fulfills $\widetilde{D}_n(G)/\sqrt{\log(n/G)}\to\infty$ as well as 
	slightly  strengthen Assumption~\ref{Ass_31_new}(b) on the signal strength to
	\begin{align*}
		\frac{\sqrt{G}}{\widetilde{D}_n(G)}\min_{\substack{k_{j-1,n}-G\le k\le k_{j-1,n}-\varepsilon\,G \\ k_{j,n}-(1-\varepsilon)\,G\le k\le k_{j,n}}}\|\boldsymbol{\widehat{\theta}}_{k+1,k+G}- \boldsymbol{\theta}_j\|\overset{P}{\longrightarrow}\infty.
	\end{align*}
\end{remark}

We are now ready to prove a proposition, from which we will derive a consistency theorem for the MOSUM segmentation including first rates of convergence.
\begin{proposition}\label{propositionsprobminmaxwald}
	Let Assumptions~\ref{as.bandwidth} on the bandwidth hold, in addition to
	Assumptions \ref{as.invariancemulti}, where (b) needs to hold with $\bt=\bt_j$, $j=1,\ldots,q+1$, (as defined in Assumption~\ref{aswaldaltnullseg}) in addition to Assumptions~\ref{aswaldaltnullseg} and \ref{Ass_31_new}.
	\begin{enumerate}[(a)]
		\item Then, 	for $D_n(\alpha_n,G)$ as in \eqref{eq_neu_threshold} fulfilling Assumption~\ref{as.siglevel}		\begin{align*}
				(i)\quad &P\left(\max_{j=1,\ldots,q+1}\max_{k_{j-1,n}+G\le k\le k_{j,n}-G} T^{(1)}_{k,n}(G)<D_n(\alpha_n,G)\right)\rightarrow 1,\\
				(ii)\quad &P\left(\min_{j=1,\ldots,q}\,\min_{|k-k_{j,n}|\le (1-\varepsilon)\,G}
	T^{(1)}_{k,n}(G)		\geq D_n(\alpha_n,G)\right)\rightarrow 1.
			\end{align*}
	\item Furthermore, the statements of (a) remain true if the long-run covariance matrix  \(\boldsymbol{\Gamma}_k\) is replaced by estimators \(\widehat{\boldsymbol{\Gamma}}_{k,n}\) satisfying Assumption \ref{asaltestcovmatwald}.
\end{enumerate}
\end{proposition}

The following theorem shows, that the number of change points are estimated consistently and that the distance from each change point to the closest estimator is smaller than $G$.

\begin{theorem}\label{theorem.consistencyqwald}
Let the assumptions of Proposition \ref{propositionsprobminmaxwald} hold (either with known or estimated long-run covariance matrix).
Then, as $n\to\infty$,
\begin{align*}
	P\left(\widehat{q}_n^{(1)}=q, \max_{1\leq j\leq q}
	\left|\widehat{k}^{(1)}_{j,n}-k_{j,n}\right|< G\right)\rightarrow 1,
\end{align*}
where $\widehat{q}_n^{(1)}$ is the estimated number of change points based on the MOSUM-Wald statistics and $\widehat{k}^{(1)}_{j,n}$ are the corresponding (ordered) change point estimators.
\end{theorem}

As a special case it follows:

\begin{remark}\label{remarkconsistencysacorcpsestwald}
	In the classical multiple change point situation, where $k_{j,n}=\lfloor \lambda_j\, n\rfloor$ for $0<\lambda_1<\ldots<\lambda_q<1$  and $\widehat{\lambda}^{(1)}_{j,n}=\widehat{k}^{(1)}_{j,n}/n$, it holds
	\begin{align*}
		\max_{j=1,\ldots,\min(q,\widehat{q}^{(1)}_n)}\left|\widehat{\lambda}_{j,n}-\lambda_{j,n} \right|=O_P\left( \frac Gn \right)=o_P(1).
	\end{align*}
\end{remark}

\subsection{MOSUM-Score Procedure}\label{chapterscore}
The main disadvantage of the MOSUM-Wald procedure is its large computational complexity in non-linear problems where $n-G$ estimators need to be calculated via numerical optimization. Additionally, the use of numerical procedures can cause  statistical problems if distant parameters describe very similar models (corresponding to many local maxima for the numerical optimization). A prominent example of this type are neural networks.The MOSUM-score procedure on the other hand does not suffer from this problem but may  (depending on the underlying model and choice of estimating function) not detect all changes if only applied with a single inspection parameter.

\subsubsection{Detectability}\label{sectiondetectabibilty}
A change in the parameter vector $\bt$ can only be detected or localized by the MOSUM-score statistics if it causes a change in the expectation of the transformed series $\boldsymbol{H}(\mathbb{X}_i,\tbt)$ for the chosen inspection parameter $\tbt$.
Denote the corresponding set and its cardinality by
\begin{align}\label{as.alt.dectest2}
	&\tilde{Q}=\tilde Q(\tbt)=
	\left\{
          1\le j\le q\,\Big|\, \E\boldsymbol{H}(\mathbb{X}_{1}^{(j)},\boldsymbol{\widetilde{\theta}})\neq \E\boldsymbol{H}(\mathbb{X}_{1}^{(j+1)},\boldsymbol{\widetilde{\theta}}) 
	  \right\},
	  \qquad\notag\\&
	  \tilde{q}=\tilde{q}(\tbt)=|\tilde{Q}|.
\end{align}
In order to formulate the below results it is helpful to relabel detectable changes $k_{j,n}$ with $j\in \tilde{Q}$ by $\tilde{k}_{j,n}=\tilde{k}_{j,n}(\tbt)$, $j=1,\ldots,\tilde q$, where $\tilde{k}_{1,n}<\ldots<\tilde{k}_{\tilde{q},n}$.
\begin{figure}
	\caption{The plots give two examples where $\tilde{q}\neq q$, i.e.\ not all changes are detectable. The upper panel shows a simulated time series, while the lower panel shows the corresponding signal of the MOSUM-score statistic with the global median as inspection parameter and the estimating function of the median. The change points are indicated by the red vertical lines.}\label{figuredetectptob} 
\begin{center}\vspace{-3ex}
\includegraphics[width=0.9 \textwidth]{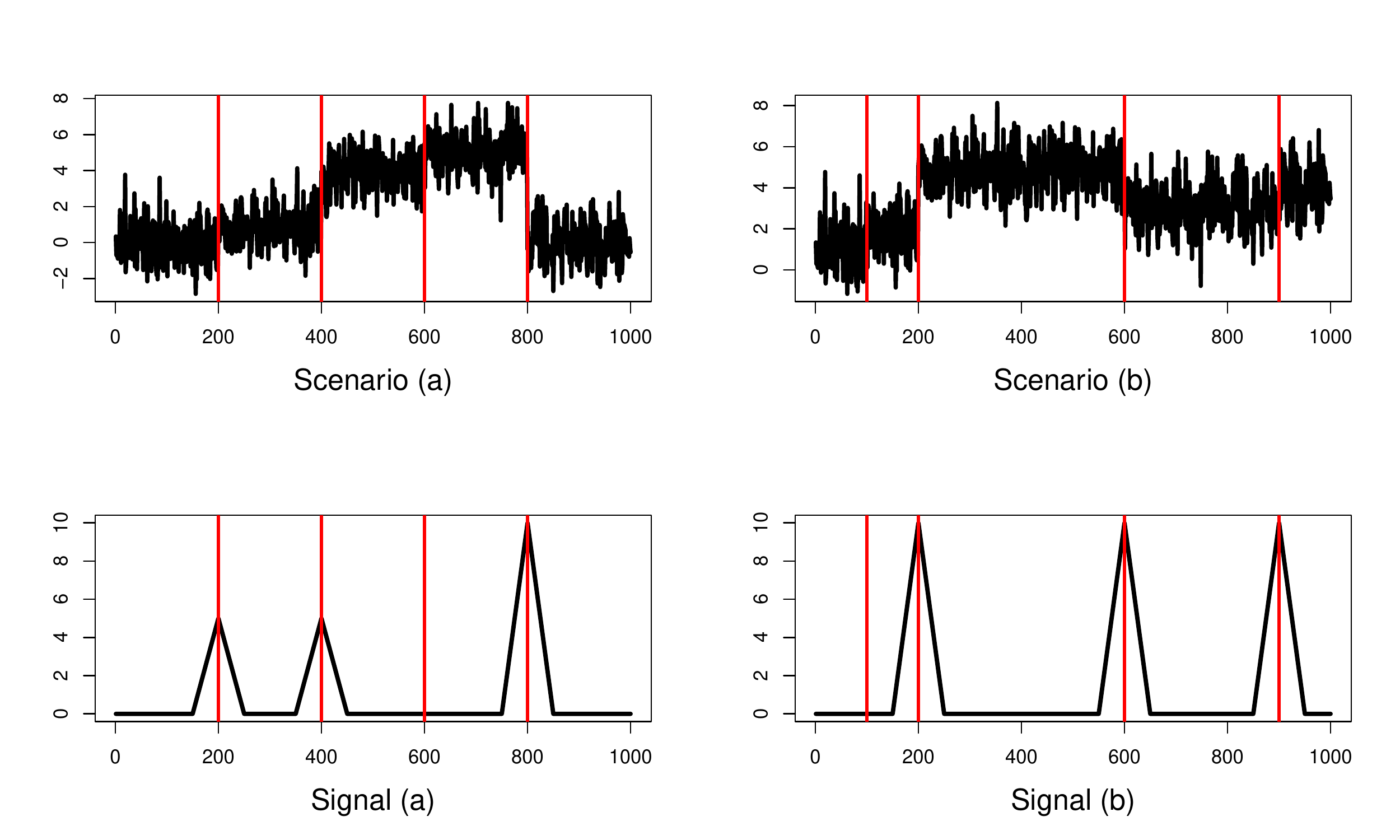}\end{center}\end{figure}

Clearly, the number of detectable changes depends on both the choice of estimating function and the inspection parameter. 
Figure \ref{figuredetectptob} gives an example for $\tilde q\neq q$, where the estimating function for the median, i.e.\ $H(x,\mu)=\sgn(x-\mu)$, was used with the global median as inspection parameter. Using a smooth strictly monotone approximation of this estimating function (compare Section~\ref{examplemeanchange}) makes all changes detectable theoretically but still leads to power loss for some changes, which can be remedied by using different inspection parameters in addition to some post-processing (see Remark~\ref{rem_detect} below). This effect is illustrated by means of a simulation study as well as a data example in Section~\ref{section_det_sim} below.

For the classical multiple change point situation, where $ k_{j,n}=\lfloor  \lambda_j\, n\rfloor$ for $0=\lambda_0< \lambda_1<\ldots< \lambda_{ q}<\lambda_{q+1}=1$ a canonical data-driven inspection parameter is $\tbt_n=\widehat{\boldsymbol{\theta}}_{1,n}$ defined as the solution of $\sum_{i=1}^n\bs{H}(\mathbb{X}_i,\bt)\overset{!}{=}\bs{0}$. 

Theorem~\ref{theoremrootnconsisalt} below shows that under mild conditions this data-dependent inspection parameter fulfills the required assumptions with $\tbt_{0,1}$ as the unique solution of  $\sum_{j=1}^{q+1}(\lambda_{j}-\lambda_{j-1})\E \bs{H}(\mathbb{X}_1^{(j)},\bt)\overset{!}{=}\mathbf{0}$.
 \begin{lemma}\label{lemmaproblemdetect}
	Let $\bt_j$ denote the unique zero of $\E\boldsymbol{H}(\mathbb{X}_1^{(j)},\boldsymbol{\theta})$, $j=1,\ldots,q+1,$ with $\bt_j\neq \bt_{j+1}$ for all $j=1,\ldots,q$ and $q\ge 1$.
	\begin{itemize}
		\item[(a)] Then, the MOSUM-score procedure with detection parameter $\tbt_{0,1}$ (or a corresponding data driven version) detects at least one change, i.e. $\widetilde{q}\geq 1$.
\item[(b)] If there are only two distinct regimes then all changes are detectable, i.e. $\widetilde{q}=q$.
\end{itemize} 
\end{lemma}

\begin{remark}\label{rem_detect}In particular, Lemma~\ref{lemmaproblemdetect} shows that at least one change point is detectable if the global estimator (computed on the whole sample) is used as inspection parameter. Consequently, by recursively applying this procedure on detected segments, all changes are detected with a much smaller computational burden than for the MOSUM-Wald procedure. An empirical illustration based on simulated data and a well-log data set is given in Section~\ref{section_det_sim} below. This is of particular interest if the procedure is used for candidate generation to be combined with a pruning step as suggested by \cite{ChoKirch2018}. Some first considerations and results in that direction can be found in Chapter~5 of ~\cite{Reckruehm}.
\end{remark}

\subsubsection{Consistency}
The MOSUM-score statistics depends on the long-run covariance matrix $\bs{\Sigma}_k(\tbt)$ as in \eqref{eq_Gamma} which is typically not known. We show that it can be replaced by an estimator as long as the estimator is consistent away from changes and fulfills some weaker assumptions close to a change point. The latter is important because local estimators will typically be contaminated by a close-by change point. For many models (other than mean changes) $\bs{\Sigma}_k(\tbt)$ differs from one segment to the next, such that there is
 is no alternative to local estimation.
\begin{assumption}\label{as.alt.varianceestimator}
	Let $\boldsymbol{\widehat{\Sigma}}_{k,n}$ satisfy 
\begin{enumerate}[(a)]
	\item $\max\limits_{k_{j-1,n}+G\le k\le k_{j,n}-G}\left\lVert \boldsymbol{\widehat{\Sigma}}^{-1/2}_{k,n}-\boldsymbol{\Sigma}_k(\tbt)^{-1/2}\right\rVert=o_P\left(\log(n/G)^{-1}\right),$ \\[1mm]for any $j=1,\ldots,q+1$,
	\item For any $j=1,\ldots,q$ it holds	\begin{align*}
		&\sup_{|k-k_{j,n}|<  G}\left\lVert\boldsymbol{\widehat\Sigma}_{k,n}^{1/2}\right\rVert<\infty,\qquad \sup_{|k-k_{j,n}|< G}\left\lVert\boldsymbol{\widehat\Sigma}_{k,n}^{-1/2}\right\rVert<\infty.
\end{align*}
	\end{enumerate} \end{assumption}

\begin{remark}\label{rem_cov_score}
	As in Remark~\ref{rem_cov_wald} it can be difficult or numerically expensive to use an estimator that is uniformly consistent away from change points. This is even more important in combination with the MOSUM-score statistics that is  mainly designed to ease the computational burden of the procedure (see also Section~\ref{section_simulation}). 
	Here, too, the below results remain true as long as one uses  a threshold $\widetilde{D}_n(T)$ with $\widetilde{D}_n(T)/\sqrt{\log(n/G)}\to\infty$.
\end{remark}

In analogy to the MOSUM-Wald statistics we prove the following proposition.
\begin{proposition}\label{propositionsprobminmax}
	Let Assumptions~\ref{as.bandwidth} on the bandwidth hold, in addition to
	 Assumptions~\ref{as.replacemulti}, thus allowing for a sequence of inspection parameters $\tbt_n$. Furthermore, let Assumptions \ref{as.invariancemulti} hold  with  $\bt=\tbt$ (as in Assumption~\ref{as.replacemulti}) in (b).	\begin{enumerate}[(a)]
\item Then, 			\begin{align*}
		(i)\quad &P\left(\max_{j=1,\ldots,\tilde{q}+1}\max_{\tilde k_{j-1,n}+G\le k\le \tilde k_{j,n}-G} T^{(2)}_{k,n}(G)<D_n(\alpha_n,G)\right)\rightarrow 1,\\
				(ii)\quad &P\left(\min_{j=1,\ldots,\tilde q}\,\min_{|k-\tilde k_{j,n}|\le (1-\varepsilon)\,G}
	T^{(2)}_{k,n}(G)		\geq D_n(\alpha_n,G)\right)\rightarrow 1.
			\end{align*}
		\item Furthermore, the statements of (a) remain true if the long-run covariance matrix  \(\boldsymbol{\Sigma}_k(\tbt)\) is replaced by  estimators \(\widehat{\boldsymbol{\Sigma}}_{k,n}\) satisfying Assumption~\ref{as.alt.varianceestimator}.
\end{enumerate}
\end{proposition}

From this we conclude in the following theorem that the number of change points are estimated consistently and that the distance from each change point to the closest estimator is smaller than $G$.

\begin{theorem}\label{theorem.consistencyq}
	Let the assumptions of Proposition \ref{propositionsprobminmax} hold (either with known or estimated long-run covariance matrix).
Then, as $n\to\infty$,
\begin{align*}
	P\left(\widehat{q}_n^{(2)}(\tbt_n)=\tilde q(\tbt), \max_{1\leq j\leq \tilde q}	\left|\widehat{k}^{(2)}_{j,n}(\tbt_n)-\tilde k_{j,n}(\tbt)\right|< G\right)\rightarrow 1,
\end{align*}
where $\widehat{q}_n^{(2)}(\tbt_n)$ is the estimated number of change points based on the MOSUM-score statistics with inspection parameters $\tbt_n$ and $\widehat{k}^{(2)}_{j,n}(\tbt_n)$ are the corresponding (ordered) change point estimators.
\end{theorem}

\begin{remark}\label{remarkconsistencysacorcpsestscore}
	In the classical multiple change point situation, where $\tilde k_{j,n}=\lfloor\tilde  \lambda_j\, n\rfloor$ for $0<\tilde \lambda_1<\ldots<\tilde \lambda_{\tilde q}<1$  and $\widehat{\lambda}^{(2)}_{j,n}=\widehat{k}^{(2)}_{j,n}/n$, it holds
	\begin{align*}
		\max_{j=1,\ldots,\min(\tilde q,\widehat{q}^{(2)}_n)}|\widehat{\lambda}_{j,n}-\tilde \lambda_{j,n}| =O_P\left( \frac Gn \right)=o_P(1).
	\end{align*}
\end{remark}

\subsubsection{Localization Rates}\label{subsectionimprovedconvrates}
Localization rates quantify 
 the distance between true and estimated change points. For a bounded number of change points the minimax optimal localization rates are known to be constant, i.e. it holds $\max_{j}|\widehat{k}_{j,n}-k_j|=O_P(1)$; see e.g.\ Lemma 2 in ~\cite{wang2020univariate}. In this section we show that the MOSUM-score estimators match these minimax optimal localization rates.

 This is important because the MOSUM-score  procedure is particularly well suited as a candidate generating algorithm to be combined with a pruning step (see \cite{ChoKirch2018} for a corresponding discussion in the context of the mean change model). This is due to its computational efficiency and the fact that combined with several inspection parameters all changes can be detected. For such pruning methodology to work it is essential that the candidate generating algorithm produces estimators that are sufficiently close to the true change points.

 In this section, we work with $\widehat{k}_{j,n}^{(2)}(\tbt_n;\widehat{\boldsymbol{\Psi}}_{j,n})$ as in \eqref{eq_est_Psi}, i.e.\ for $j\in\widetilde{Q}$
\begin{align*}
\widehat{k}_{j,n}^{(2)}(\tbt_n;\widehat{\boldsymbol{\Psi}}_{j,n})
	=\argmax_{v_{j,n}\leq k \leq w_{j,n}}\M_{\tbt_n}(k)^T\boldsymbol{\widehat{\Psi}}_{j,n}^{-1}\M_{\tbt_n}(k).
\end{align*}

We need the following assumptions on the matrices $\boldsymbol{\widehat{\Psi}}_{j,n}$.
\begin{assumption}\label{ass_Psi}
$\boldsymbol{\widehat{\Psi}}_{j,n}$ is a sequence (possibly depending on the data) of symmetric positive definite matrices with
\begin{align*}
\|\boldsymbol{\widehat{\Psi}}_{j,n}^{-1}\|=O_P(1), \qquad\|\boldsymbol{\widehat{\Psi}}_{j,n}^{1/2}\|=O_P(1).
\end{align*}
\end{assumption}
Effectively, the assumption guarantees that the use of $\boldsymbol{\widehat{\Psi}}_{j,n}$ does not kill the signal nor does it explode the noise.

Furthermore, we need to slightly strengthen the previous assumptions  by additionally requiring:
\begin{assumption}\label{ass_localization_rates}
	\begin{enumerate}[(a)]
		\item With $\tbt$ as in Assumption~\ref{as.replacemulti} we require that for any $\xi_n \to\infty$ and any $j=1,\ldots,\tilde{q}_n+1$, it holds 
			\begin{align*}
				(i)\quad& \text{for }m=\tilde{k}_{j-1,n}(\tbt)+G,\;\tilde{k}_{j,n}(\tbt)-G,\; \tilde{k}_{j,n}(\tbt)\text{ that}\\
				&\max\limits_{\xi_n\leq k\leq G}\frac{1}{k}\left\lVert\sum\limits_{i=m-k+1}^m\left(\boldsymbol{H}\left(\mathbb{X}_i^{(j)},\tbt_n\right)-\boldsymbol{H}\left(\mathbb{X}_i^{(j)},\boldsymbol{\widetilde{\theta}}\right)\right)\right\rVert=o_P(1),\\
				(ii)\quad &\text{for }m=\tilde{k}_{j-1,n}(\tbt),\;\tilde{k}_{j-1,n}(\tbt)+G,\;\tilde{k}_{j,n}(\tbt)-G\text{ that}\\
				&\max\limits_{\xi_n\leq k\leq G}\frac{1}{k}\left\lVert\sum\limits_{i=m+1}^{m+k}\left(\boldsymbol{H}\left(\mathbb{X}_i^{(j)},\tbt_n\right)-\boldsymbol{H}\left(\mathbb{X}_i^{(j)},\boldsymbol{\widetilde{\theta}}\right)\right)\right\rVert=o_P(1).
			\end{align*}
		\item The following backward law of large numbers holds for any $j$ and any sequence $\xi_n\to\infty$
			\begin{align*}
				& \max\limits_{\xi_n\leq k\leq G}\frac{1}{k}\left\lVert\sum\limits_{i=G-k+1}^G\left(\boldsymbol{H}\left(\mathbb{X}_i^{(j)},\tbt\right)-\E\boldsymbol{H}\left(\mathbb{X}_i^{(j)},\tbt\right)\right)\right\rVert=o_P(1).
			\end{align*}
	\end{enumerate}
\end{assumption}

\begin{remark}\label{rem_localization_rates}
	The forward version (as in (a)(ii)) is missing in (b) because it follows directly from a strong law of large number which (after changing the probability space) follows directly from the invariance principle as in Assumption~\ref{as.invariancemulti} (b).
	However, the backward version does not follow from that assumption alone - unless the data is i.i.d. Nevertheless, most sequences will allow for strong ergodic theorems even in backward direction. For example mixing conditions as in Regularity conditions~\ref{regc_ip}(b) are symmetric in the sense that the backward sequence is also mixing with the same rates. Consequently, both a forward and backward invariance principle holds and (b) follows from that.
\end{remark}

In the following theorem we derive the minimax optimal localization rates improving upon Theorem~\ref{theorem.consistencyq}.
\begin{theorem}\label{theoremimprovedconvrates}Let the assumptions of Theorem~\ref{theorem.consistencyq} hold in addition to Assumptions~\ref{ass_Psi} -- \ref{ass_localization_rates}. Then, it holds
	\begin{align*}
		\max_{1\le j\le \min(\tilde q,\widehat{q}_n^{(2)}(\tbt_n))} \left|\widehat{k}_{j,n}^{(2)}(\tbt_n;\widehat{\boldsymbol{\Psi}}_{j,n})-\tilde{k}_{j,n}(\tbt)\right|=O_P(1).
	\end{align*}
\end{theorem}

\section{Regularity Conditions for Sufficiently Smooth Estimating Functions}\label{sectionasunderregcond}
The results of the previous sections were obtained under certain high-level assumptions, that need to be verified separately for each underlying time series structure as well as estimating function. On the other hand, many estimating functions are either sufficiently smooth or can be approximated to any degree of accuracy by sufficiently smooth functions. Therefore, we will verify these high-level assumptions exemplarily for sufficiently smooth estimating function for strongly mixing time series under moment conditions. This includes the i.i.d.\ situation as a special case.

The invariance principles of Assumption~\ref{as.invariancemulti} (b) follow from the following regularity conditions:
\begin{regc}\label{regc_ip}	For $j=1,\ldots,q+1$ 
	\begin{enumerate}[(a)]
	\item let there exist a $\tilde\nu>0$ such that $0<\E\left\lVert \boldsymbol{H}(\mathbb{X}_1^{(j)},\boldsymbol{\theta})\right\rVert^{2+\tilde\nu}<\infty$ for the $\boldsymbol{\theta}$ of interest,
	\item let	$\{\mathbb{X}_t^{(j)}\}$ be stationary and strongly mixing sequences of random vectors with a strong mixing coefficient $\alpha(\cdot)$
	satisfying $\alpha(n)=O(n^{-\beta})$ for some $\beta>1+2/\tilde\nu$ as $n\to\infty$, where $\tilde\nu$ is as (a).
	\end{enumerate}
\end{regc}

\begin{theorem}\label{theorem_ass_ip}
	\begin{enumerate}[(a)]
		\item Under Regularity Conditions~\ref{regc_ip} Assumption~\ref{as.invariancemulti} hold for the same $\boldsymbol{\theta}$ and some $\nu>0$ depending on $\beta$, $\tilde{\nu}$ and the dimension. 
		\item Additionally, Assumption~\ref{ass_localization_rates} (b) holds if Regularity Conditions~\ref{regc_ip} are fulfilled with $\boldsymbol{\theta}=\tbt$.	
		\end{enumerate}

\end{theorem}

In the following we will always assume that the bandwidth $G$ fulfils Assumption~\ref{as.bandwidth} for the above suitable $\nu$.
 
In order to prove the  necessary assumptions for the consistency results of the above MOSUM statistics the following additional regularity conditions on the estimating function are sufficient.
\begin{regc}\label{regc_wald}\begin{enumerate}[(a)]
	\item Let $\boldsymbol{\Theta}\subset\mathbb{R}^p$ be a compact parameter space and $\boldsymbol{H}=(H_1,\ldots,H_p)^T$ twice differentiable such that for some $\tilde{\nu}>0$
\begin{align*}
	&(i)\quad\E\left\lVert \nabla \boldsymbol{H}(\mathbb{X}_1^{(j)},\boldsymbol{\theta})\right\rVert^{2+\tilde\nu}<\infty\qquad\text{for all }\boldsymbol{\theta}\in\boldsymbol{\Theta},\\
	&(ii)\quad \E\sup_{\boldsymbol{\theta}\in\boldsymbol{\Theta}}\left\lVert \nabla^2 H_l(\mathbb{X}^{(j)}_1,\boldsymbol{\theta})\right\rVert^{2+\tilde\nu}<\infty, \quad \text{for all } l=1,\ldots,p,
\end{align*}
for all $j=1,\ldots,q+1$. 
\item Let $\bs{V}_{(j)}(\bt)=\E\left(\nabla\boldsymbol{H}(\mathbb{X}^{(j)}_1,\boldsymbol{\theta})\right)^T$ be a regular matrix for all $\boldsymbol{\theta}\in\boldsymbol{\Theta}$ fulfilling for all $j=1,\ldots,q+1$
\begin{align*}
	&\sup_{\bt\in\bs{\Theta}}\left\lVert\bs{V}_{(j)}(\bt)^{-1}\right\rVert<\infty.
\end{align*}
\item Let $\E\sup_{\boldsymbol{\theta}\in\boldsymbol{\Theta}}\left\lVert  \nabla H_l(\mathbb{X}^{(j)}_1,\boldsymbol{\theta})\right\rVert<\infty$, $j=1,\ldots,q+1$.
\end{enumerate}
\end{regc}
For specific examples the above assumptions such as the compactness assumption can be weakened. The  important special case of linear regression models has been discussed in detail in \cite{Reckruehm}, Chapter 3.2, under standard assumptions for linear regression.

\begin{theorem}\label{theorem_ass_wald}Let Regularity Condition \ref{regc_ip}  hold with $\boldsymbol{\theta}=\boldsymbol{\theta}_j$ as defined in Assumption~\ref{aswaldaltnullseg} identifiably unique, $j=1,\ldots,q+1$, in addition to Regularity Condition~\ref{regc_wald} (a) and (b). 	\begin{enumerate}[(a)]
			\item Then Assumptions~\ref{aswaldaltnullseg} and  Assumption~\ref{Ass_31_new} (a) hold.
			\item If additionally Regularity Condition~\ref{regc_wald} (c)  hold, we get Assumption~\ref{Ass_31_new} (b).
			\end{enumerate}

\end{theorem}

The MOSUM-score procedure requires an inspection parameter $\tbt$ which can also be data-dependent $\tbt_n$. We will show the validity of the assumptions  for situations where $\tbt_n$ is a $\sqrt{n}$-consistent estimator for some $\tbt$ in Theorem~\ref{theorem_ass_score}. 

Here, we discuss the choice $\widehat{\bt}_{a,b}$ defined as the unique zero of $\sum_{i=a}^b\bs{H}(\mathbb{X}_i,\bt)\overset{!}{=}\bs{0}$, which has advantageous properties  for detectability (compare Section~\ref{sectiondetectabibilty}) in particular if combined with a subsequent pruning step.
\begin{theorem}\label{theoremrootnconsisalt}
Let the classical multiple change point situation with $ k_{j,n}=\lfloor  \lambda_j\, n\rfloor$ for $0=\lambda_0< \lambda_1<\ldots< \lambda_{ q}<\lambda_{q+1}=1$ hold. Consider $a=\lfloor \gamma_a n\rfloor$ and $b=\lfloor \gamma_b n\rfloor$ for some $0\le \gamma_a<\gamma_b\le 1$.
Define $\tbt_{\gamma_a,\gamma_b}$ as the unique solution of  $\sum_{j=1}^{q+1}(\min(\lambda_{j},\gamma_b)-\max(\lambda_{j-1},\gamma_a))_+\,\E \bs{H}(\mathbb{X}_1^{(j)},\bt)\overset{!}{=}\mathbf{0}$, where $x_+=\max(x,0)$. 
Furthermore, let Regularity Conditions~\ref{regc_ip} hold with $\bt=\tbt_{\gamma_a,\gamma_b}$ as well as ~\ref{regc_wald} (b) and (c).
Then,
\begin{align*}
	\sqrt{n} \left(\widehat{\bt}_{a,b}- \tbt_{\gamma_a,\gamma_b} \right)=O_P(1).
\end{align*}
\end{theorem}

Finally, we show the validity of the necessary assumptions for the MOSUM-score procedure under suitable regularity conditions on the estimating function.
\begin{theorem}\label{theorem_ass_score}
	Let $\tbt_n$ and $\tbt$ with $\sqrt{n} (\tbt_n-\tbt)=O_P(1)$ in addition to Regularity Condition \ref{regc_ip} with $\boldsymbol{\theta}=\tbt$ hold. Additionally, let Regularity Condition~\ref{regc_wald} (a) and (c) hold. Then
Assumption~\ref{as.replacemulti} as well as \ref{ass_localization_rates} (a) hold.
\end{theorem}

\section{Empirical comparison}\label{section_simulation}
For the mean change situation based on the sample mean as in Section~\ref{examplemeanchange} extensive simulation studies can be found in \cite{KirchMuhsal,ChoKirch2018,ChoKirchMeier}.
In this section, we want to give a proof of concept beyond the mean by including a variety of examples in the simulation study, where we focus on aspects that differ from the mean situation.
In Section~\ref{section_det_sim} we aim at sheding some additional light on the detectability issue discussed in Section~\ref{sectiondetectabibilty} using the example of a mean change in combination with the median-like estimator.
In Sections~\ref{sec_linear} and \ref{sec_Poisson} we focus on differences between the MOSUM-Wald and MOSUM-score statistics in the empirical performance as well as their computation time. In particular, this includes the example of linear regression in Section~\ref{sec_linear}, where the estimator is known analytically, as well as the example of Poisson autogression in Section~\ref{sec_Poisson}, where this is not the case and all estimators need to be obtained numerically.

All simulations are based on $1000$ repetitions unless otherwise stated, $\varepsilon$ as in \eqref{MOSUMcond3} is set to $0.2$ and the threshold from Section~\ref{sectionthreshold} with $\alpha_n=5\%$ is used.

\subsection{Median-like Estimator and Detectability}\label{section_det_sim}

\begin{table}
\begin{center}\begin{tabular}{|c||c|c|c|c|}
\hline
   $G$ &   $[80,120]$& $[180,220]$ & $[580,620]$ & $[880,920]$ \\
 \hline \hline
 \multicolumn{5}{|c|}{Inspection Parameter: Sample median $\hat{\mu}_{1,1000}$}\\
 \hline
$20$        &   0.019   &  1.000 &    0.935  &   0.142\\
$50$          &    0.343  &   1.000   &  1.000  & 0.665\\
  \hline  \hline
  \multicolumn{5}{|c|}{Inspection Parameter: Sample median $\hat{\mu}_{1,200}$}\\
\hline
$20$     &     0.158     & 0.985 &   0.380  &   0.026  \\
$50$      &    0.659     & 1.000  &   0.999   &  0.404  \\
  \hline 
 \end{tabular}\end{center}
\caption{Percentage of simulations with a change point estimator in the given intervals ('detection rate')  based on the MOSUM-score statistic with the median-like estimating function}
\label{table_median}
 \end{table}

In Section \ref{sectiondetectabibilty} the detectability issue of the MOSUM-score procedure was discussed in detail. Here, we shed some additional light onto the problem by a corresponding empirical study and the analysis of well-log data where we use the median-like estimating function as discussed in Section~\ref{examplemeanchange} with the sample median (based on different stretches of the original data)
as inspection parameter. If instead the median-like estimator is used as an inspection parameter, the detection rates as given in Table~\ref{table_median} only differ in the third digit. However, using the sample median as implemented in \texttt{R} instead of the median-like estimator (implemented using the \texttt{uniroot}-function from the \texttt{R}-package \texttt{stats}) saved around 1/3 of computation time in our other-wise identical simulation.

The procedure is implemented using the \texttt{mosum} function from the \textsf{R}-package \texttt{mosum}  applied to the transformed data sequence $H(X_t,\widehat{\bt}_n)$. In particular, the variance is estimated empirically with the mosum-window estimator that is the default in the package (see \cite{ChoKirchMeier}). Thus, the signal is effectively given by the difference in the expectation of the transformed data $\E H(X_t,\tbt)$  with the noise corresponding to the distribution of $H(X_t,\tbt)$ where $\tbt$ is the limit of $\widehat{\bt}_n$.

To elaborate we consider a noise sequence of $1000$ independent standard normally distributed random variables which are then centered around a mean sequence of $1,2,5,3,4$ with change points at $100, 200, 600$ and $900$. Consequently, the global median will be somewhere between $3$ and $4$. As such the second (jumping from mean $2$ to mean $5$) and to a slightly lesser degree the third change point (jumping from $5$ to $3$) should be well visible to our statistic applied with the global median, while the last one (jumping from $3$ to $4$) should be less visible. This is also confirmed by the empirical detection rates as reported in Table~\ref{table_median}.  The first change point should be almost  invisible as it jumps from $1$ to $2$ with 
only few observations exceeding the global median. While the detection rates in Table \ref{table_median} are indeed lower than for the other change points, they are  well above the $5\%$ threshold with detection rates of approximately $20\%$ (for $G=20$) and even above $80\%$ (with $G=50$). The reasons are twofold: First of all, more observations from the second regime can be expected to cross the threshold than for the first one. Secondly, with the median-like estimation function, the data transformation is strictly monotonic such that any mean difference remains theoretically detectable with any inspection parameter. Furthermore, the transformation that dampens the signal (if both means are on the same side of the inspection parameter) will also dampen the noise (i.e.\ variance) of the signal.  As soon as the sample median of the first $200$ observations is used as an inspection parameter, which will be between $1$ and $2$,  the detection rates for the first change point increases while the ones for the last change point decreases illustrating Remark~\ref{rem_detect}.

\subsubsection*{Analysis of well-log data}

One of the advantages of using the median in combination with the median-like estimating function is the fact that they are robust with respect to outliers. Indeed, recently, the distinction between outliers and change points has received some increased attention. For example, \cite{knoblauch2018doubly}, \cite{fearnhead2019changepoint} and \cite{li2021adversarially} present robust change point procedures and illustrate their usefulness using well-log data (see Figure~\ref{figure_well_log}~(a)): This geophysical data is collected from a probe being lowered into a bore-hole, where change points occur when the probe moves from one rock strata to another. The data notably contains several outliers, which have been removed before the change point analysis in earlier works such as \cite[Section 5.7.2]{ruanaidh1996numerical}, where the data was first discussed, or \cite{fearnhead2003line}, \cite{fearnhead2006exact}, \cite{adams2007bayesian},\cite{wyse2011approximate}, \cite{ruggieri2016exact}. 

\begin{figure}[tb]
	\centering
	\subfloat[Well-log data with detected change points]{
\includegraphics[width=\textwidth,trim=1.5cm 1.6cm 1.5cm 2.2cm]{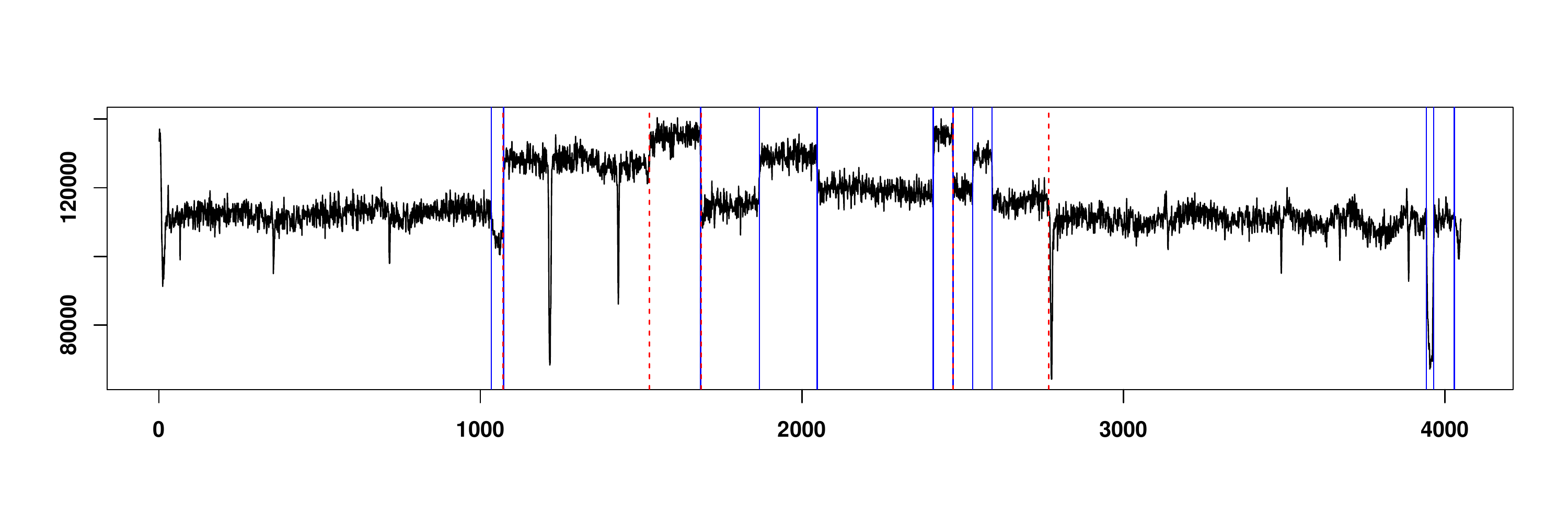}
	}\\
	\subfloat[MOSUM-Score statistic with global sample median as inspection parameter]{
\includegraphics[width=\textwidth,trim=1.5cm 1.6cm 1.5cm 1.8cm]{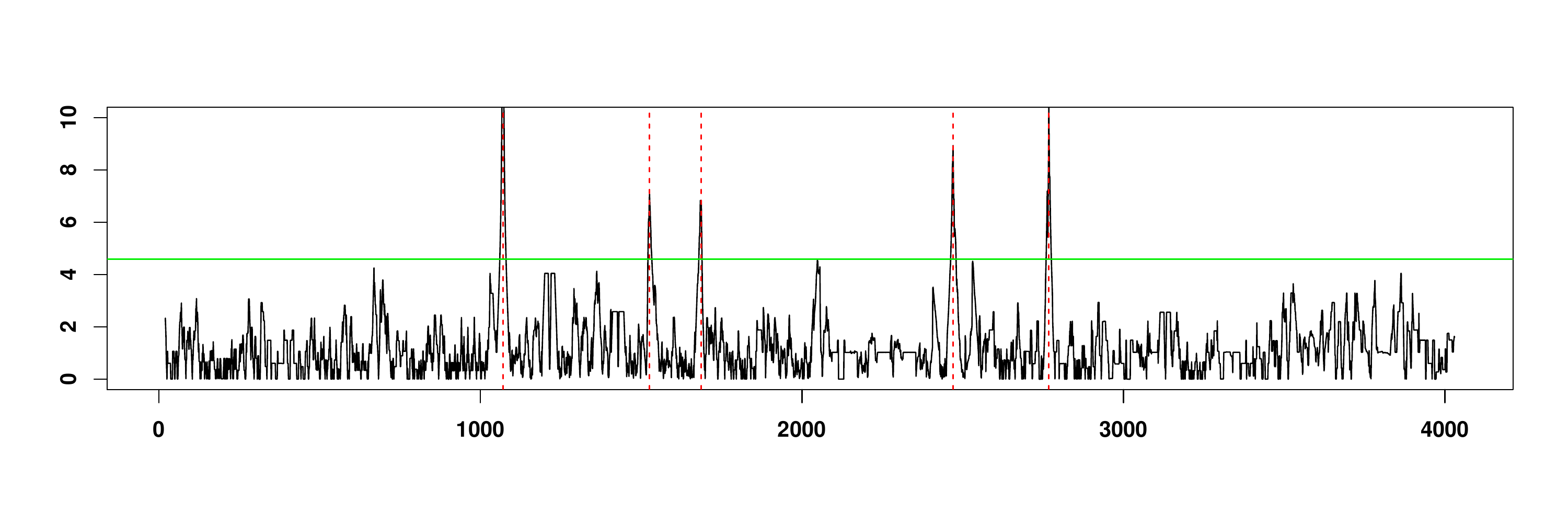}
	}\\
	\subfloat[MOSUM-Score statistic with median between $1070$ and $2767$ as inspection parameter]{
\includegraphics[width=\textwidth,trim=1.5cm 1.6cm 1.5cm 1.8cm]{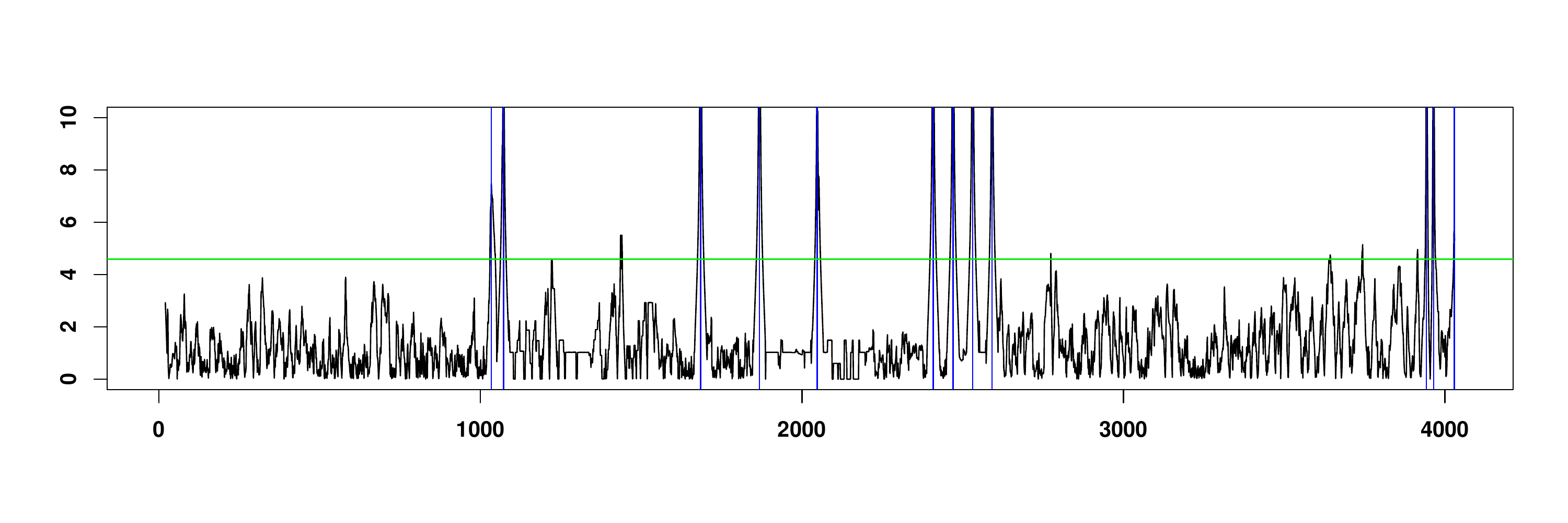}
	}
	\caption{Well-log data with segmentation obtained from the median-like MOSUM-Score statistic with different inspection parameters. The vertical dashed red lines give the detected change points using the global median as inspection parameter, while the vertical solid blue lines give the detected change points using as an inspection parameter the median of the stretch starting at $1070$ and ending at $2767$.}
	\label{figure_well_log}
\end{figure}

Using the median-like estimating function in combination with the sample median yields a robust procedure such that we can use this methodology on the original data set including the outliers. We use a bandwidth of $G=20$ and first apply the procedure using the global median as inspection parameter (see Figure~\ref{figure_well_log}~(b)). This results in the detection of 5 change points leaving some more undetected (as predicted by the discussion in Section~\ref{sectiondetectabibilty}). In a second step, we repeat the procedure using now the sample median between the first (at $1070$) and last (at $2768$) detected change point as  inspection parameter (see Figure~\ref{figure_well_log}~(c)). This results in the detection of 9 additional change points, which together result in a very reasonable segmentation of the data (see Figure~\ref{figure_well_log}~(a)). Out of the three change points detected by both inspection parameters in one case the estimator was exactly at the same location (at $2470$), while the other two were shifted by 2 time points (one in either direction). This further emphasizes the need for suitable post-processing methodology that can deal with this type of duplicate estimators obtained from different inspection parameters or different bandwidths (see also Remark~\ref{rem_detect}). Using additionally inspection parameters obtained from the data up to the first respectively starting at the last change point (detected with the global median at $1070$ respectively $2768$) does not yield any additional change point estimators but some more duplicated ones ($8$ that are exactly equal, and $6$ slightly shifted). 

Concerning the last three detected change points in Figures~\ref{figure_well_log}~(a) and (c), the question may arise as to whether these are truly change points or in fact outliers. Indeed, the stretch between $3942$ and $3965$ as been removed by some of the previous works as outliers but not all of them. The distinction between serial outliers and change points depends on the length of the 'outlier' stretch. While some authors  considered this stretch of $23$ very small observations as two change points, others have considered them serial outliers. The MOSUM procedure  is robust with respect to classical outliers that can occur in both directions in an isolated way. Longer (with respect to the bandwidth $G$) stretches of serial outliers are in some sense indistinguishable from segments of length of order $G$ and thus may be detected by the procedure: While the breakdown point of the median is at $50\%$,   smaller stretches of outliers (in the same direction) of length $\gamma G$, $0<\gamma<0.5$, will lead to the situation that the sample median of the window of $G$ observations (including the $\gamma G$ serial outliers) will estimate the $(0.5\pm 0.5\, \gamma/(1-\gamma))$-quantile instead which typically will differ from the median of the neighboring stretch that does not contain any serial outliers. This also explains the last change point in Figures~\ref{figure_well_log}~(a) and (c) because there the 'outlier' stretch (if one decides to classify it as such) has a length of around $10$ and we use a bandwidth of only~$20$. 
On the other hand a suitable choice of bandwidth $G$ in combination with the choice of $\varepsilon$ in \eqref{MOSUMcond3} will make the procedure robust with respect to sufficiently short outlier stretches. The latter effect can e.g.\ be seen in Figure~\ref{figure_well_log}~(c), where some of the other outlier stretches lead to significant points in the MOSUM score statistic that are sufficiently isolated to not be considered change points by the procedure due to the choice of $\varepsilon=0.2$ in \eqref{MOSUMcond3}.

With larger bandwidths the long stretches at the beginning and the end of the sequence are also segmented to some extend, which is consistent with observations made by some previous authors (even after outlier removal) because these stretches seem not to be mean-stationary which will cause location segmentation procedures to fit a step-function.

\subsection{Linear Regression}\label{sec_linear}
In this section, we will explore the difference in performance between the MOSUM-Wald and MOSUM-score methodology when applied to linear regression where we use the usual least squares methodology that requires no numerical approximation.

To this end, we consider a time series $Y_i=\boldsymbol{X}_i^T\boldsymbol{\beta}+\varepsilon_i$, $i=1,\ldots,n,$ of length $n=1000$ with exogenous regressors $\boldsymbol{X}_i=(1,X_{i,1},X_{i,2})^T$ with $X_{i,1} \sim N(1,1)$ and $X_{i,2} \sim N(2,1)$ and i.i.d.\ standard normal errors. We include three change points at $200, 500$ and $800$ and the regression coefficients $\boldsymbol{\beta}_1=(1,2,2)^T, \boldsymbol{\beta}_2=(1,1,2)^T, \boldsymbol{\beta}_3=(2,1,2)^T$ and  $\boldsymbol{\beta}_4=(2,1,1)^T$. We use the global least-squares regression estimator as inspection parameter $\hbb_{1,n}$ for the MOSUM-score statistic.

As discussed in Remarks~\ref{rem_cov_wald} and \ref{rem_cov_score} we can make minimal assumptions on the covariance estimators theoretically, but the small-sample performance crucially depends on this estimator.

\begin{table}\begin{center}\begin{tabular}{|l||c|c|c|c|c||c|c|c|}
\hline
&\multicolumn{5}{|c||}{\textbf{Estimated number $\widehat{q}$}}&\multicolumn{3}{|c|}{\textbf{Detection rate }}    \\
  $G$ &  $\leq 1$ & $2$ & $3$ & $4$ & $\geq5$ & $200$& $500$ & $800$ \\
 \hline \hline
 \multicolumn{9}{|c|}{MOSUM-score with \emph{S-global} covariance estimator}\\
 \hline
 $50$ & 0.484  &0.489 &0.027 &0 &0&0.494 & 0.027& 0.993\\
 $100$ & 0.003&0.468 & 0.518&0.011 & 0 & 0.969 & 0.515& 0.999\\
  \hline \hline
  \multicolumn{9}{|c|}{MOSUM-score  with \emph{S-local} covariance estimator}\\
 \hline
 $50$ & 0.110  &    0.502    &  0.353  &    0.034   &       0.001 &  0.804  & 0.430  &  1.000\\
 $100$ & 0    &  0.049   &   0.918   &   0.033  &        0 &  0.985  & 0.917 &   1.000 \\
 \hline \hline
\multicolumn{9}{|c|}{MOSUM-Wald with \emph{W-local} covariance estimator}\\
 \hline
 $50$ & 0.018  &0.445 &0.501 &0.035 &0.001& 0.963 & 0.539 & 1.000 \\
 $100$ & 0 &0.030 & 0.945 &0.025 & 0& 0.998 & 0.938& 1.000 \\
  \hline
 \end{tabular}\end{center}
 \caption{Number of estimated change points and detection rate for all three change points (i.e.\ percentage of simulations with a change point estimator in the interval $[k_{j,n}-20,k_{j,n}+20]$) for the various scenarios in the linear regression example. }
\label{table_lin_reg}
 \end{table}

As covariance estimators we use $ \boldsymbol{\widehat{\Sigma}}^{(j)}=\widehat{v}_{n}^{(j)}\,\frac{1}{n}\sum_{i=1}^n\boldsymbol{X}_i\boldsymbol{X}_i^T$ with
\begin{align*}
	&\widehat{v}_{n}^{(1)} =\frac{1}{n-1}\sum_{i=1}^n(Y_i-\bs{X}_i^T\hbb_{1,n})^2,\tag{\emph{S-global}}\\
&\widehat{v}_{n}^{(2)} = \frac{1}{2G}\left(\sum_{i=k-G+1}^{k}(\hat{\epsilon}_{i}-\bar{\epsilon}_{k-G+1,k})^2+\sum_{i=k+1}^{k+G}(\hat{\epsilon}_{i}-\bar{\epsilon}_{k+1,k+G})^2\right),\\
& \text{with }\hat{\epsilon}_{i}:=Y_i-\bs{X}_i^T\hbb_{1,n} \text{ and }\bar{\epsilon}_{l,u}:=\frac{1}{u-l+1}\sum_{i=l}^u\hat{\epsilon}_{i}, \tag{\emph{S-local}}\\
&\widehat{v}_{n}^{(3)} =\frac{1}{2G}\left(\sum_{i=k-G+1}^{k}(Y_i-\bs{X}_i^T\hbb_{k-G+1,k})^2+\sum_{i=k+1}^{k+G}(Y_i-\bs{X}_i^T\hbb_{k+1,k+G})^2\right).\tag{\emph{W-local}}
\end{align*}

We use the first two estimators  with the MOSUM-score statistics while the last one is only used with the Wald statistics where the estimators $\hbb_{t+1,t+G}$ are already available. We do not use this estimator with the MOSUM-score statistic because their usage cancels the computational advantage of the MOSUM-score statistic over the MOSUM-Wald statistic. 

The first two estimators $ \boldsymbol{\widehat{\Sigma}}^{(j)}$, $j=1,2$, are  only consistent in the no-change situation but do not fulfill assumption~\ref{as.alt.varianceestimator}(a) in the presence of change points. Instead of using a threshold as  in Remark~\ref{rem_cov_score}, in this simulation study, we stick to the threshold as  in \eqref{eq_neu_threshold} for all methods.

Table~\ref{table_lin_reg} gives
the estimated number of change points as well as the detection rates for all three change points in the various settings. For a bandwidth of $50$ the MOSUM-Wald procedure (with \emph{W-local}) outperforms the MOSUM-score procedure (with \emph{S-local}), while both procedures achieve a similar performance for bandwidth $100$. The MOSUM-score procedure with the global estimator $(\emph{S-global})$ is clearly inferior to both competing methods (in particular for  smaller bandwidth and the second change point) emphasizing again the importance of the choice of the covariance estimator for the small sample performance.

 In terms of computation time (for $G=n^{2/3}$) the MOSUM-score  clearly outperforms  the MOSUM-Wald statistics: Even with the local estimator  (\emph{S-local}), the MOSUM-score statistic was roughly $22$ times faster in our simulations than the MOSUM-Wald statistic  (with \emph{W-local}) for a time series of length $1000$ (an average (out of $100$ runs) of $0.03$ seconds as compared to $0.66$ seconds). For a length of $8000$ it was already more than $31$ times faster ($0.22$ versus $6.91$ seconds).
The numbers only give a qualitative idea as we did not optimize any of the procedures with respect to computation time but merely used the \textsf{R}-function \texttt{rollsum} (from the \textsf{R}-package \texttt{zoo} \cite{zooarticle}) to calculate the MOSUM-statistics, where the local estimators for the covariance matrices are implemented naively with a loop and the local regression parameter for the Wald-statistics is calculated with the \texttt{lm}-function.

More simulation results including the false alarm rate in the no-change situation, the results for other bandwidths and covariance estimators and more information on computing times  can be found in ~\cite{Reckruehm}, Section~4.1.

\subsection{Poisson Autoregressive Model}\label{sec_Poisson}
In this section we consider the Poisson autoregressive model with the estimation function corresponding to the partial likelihood as in Section~\ref{examplepoiautoreg}. Compared to the previous section this includes two additional difficulties: First, due to the serial dependence of the data the true scaling of the procedures depends on the long-run covariance rather than the covariance matrix. This time-dependency is also the reason behind the larger bandwidths compared to the previous section. Secondly, there is no analytical solution to the estimating procedure such that numerical methods are required which as expected will greatly increase computation time for the MOSUM-Wald procedure.

\begin{table}
\begin{center}\begin{tabular}{|l||c|c|c|c|c||c|c|c|}
\hline
&\multicolumn{5}{|c||}{\textbf{Estimated number $\widehat{q}$}}&\multicolumn{3}{|c|}{\textbf{Detection rate}}    \\
$G$ &  $\leq 1$ & $2$ & $3$ & $4$ & $\geq5$ & $250$& $500$ & $750$ \\
 \hline \hline
\multicolumn{9}{|c|}{MOSUM-score procedure with  $\hbt_{1,1000}$}\\
 \hline
 $80$ & 0.619 &0.288 &0.063 &0.028 & 0.002&0.713&0.135&0.242 \\
 $150$  &0.056 &0.321 &0.449 &0.137 & 0.037 &0.921&0.583&0.623\\
   \hline \hline
\multicolumn{9}{|c|}{MOSUM-score procedure with $\hbt_{300,700}$}\\
 \hline
 $80$ & 0.100 &0.397 &0.300 &0.143 & 0.060&0.936&0.199&0.734 \\
 $150$ &0.018 &0.162 &0.596  &0.194  & 0.030&0.919&0.724&0.742 \\
  \hline \hline
\multicolumn{9}{|c|}{MOSUM-Wald procedure}\\
 \hline
 $80$ & 0.069 &0.295 &0.373 &0.199 & 0.064&0.890&0.603&0.645 \\
 $150$ &0.001 &0.040 &0.629  &0.261  & 0.069&0.896&0.809&0.803 \\
  \hline 
 \end{tabular}\end{center}
 \caption{Number of estimated change points and detection rate  for all three change points (i.e.\ percentage of simulations with a change point estimator in the interval $[k_{j,n}-20,k_{j,n}+20]$)  for the various scenarios in the Poisson autoregressive example.
 }
\label{table_Poisson}
 \end{table}

We consider a time series of length $n=1000$ with three change points at times $250$, $500$ as well as $750$ with the paramaters $\bt_1=(1,0.5)^T$, $\bt_2=(2.5,0.5)^T$, $\bt_3 =(2.5,0.2)^T$ as well as $\bt_4=(1,0.5)^T$.

 For the MOSUM-score statistics we use the global (partial) maximum likelihood estimator as inspection parameter as well as the one based on the observations between time point $300$ and $700$ (compare also Remark~\ref{rem_detect}). To estimate the covariance matrix we use the following local estimator
 \begin{align*}
&	\widehat{\bs{\Sigma}}_{k,n}\\
&= \frac{1}{2G}\sum_{i=k-G+1}^k\left(\bs{H}(\mathbb{Y}_i,\hbt_{1,n})-\overline{\bs{H}}_{k-G+1,k}\right)\left(\bs{H}(\mathbb{Y}_i,\hbt_{1,n})-\overline{\bs{H}}_{k-G+1,k}\right)^T
\\*& \quad+ \frac{1}{2G}\sum_{i=k+1}^{k+G}\left(\bs{H}(\mathbb{Y}_i,\hbt_{1,n})-\overline{\bs{H}}_{k+1,k+G}\right)\left(\bs{H}(\mathbb{Y}_i,\hbt_{1,n})-\overline{\bs{H}}_{k+1,k+G}\right)^T,\notag
\end{align*}
where $\overline{\bs{H}}_{l,u}$ denotes the sample mean of $\bs{H}(\mathbb{Y}_l,\hbt_{1,n}),\ldots, \bs{H}(\mathbb{Y}_u,\hbt_{1,n})$. 
However, this estimator does not take the dependence into account and as such estimates the covariance rather than the long-run covariance matrix, such that Assumption~\ref{as.alt.varianceestimator}(a) is not fulfilled. As before we use a threshold as in \eqref{eq_neu_threshold} despite Remark~\ref{rem_cov_score}.

Motivated by \cite{weiss2010inarch}, equations (7) and (8), we use the  estimator\\ $\widetilde{\bs{\Gamma}}_{k,n}^{-1}=\frac{1}{2}\left(\widetilde{\bs{\Gamma}}_{k-G+1,k}^{-1}+\widetilde{\bs{\Gamma}}_{k+1,k+G}^{-1}\right)$
in the MOSUM-Wald procedure, where
\begin{align*}
&\widetilde{\bs{\Gamma}}_{l,u}^{-1}=\frac{1}{G}\sum_{i=l}^u \frac{1}{(\bs{Y}_{i-1}^T\hbt_{l,u}^{ML})^2}\left(\begin{array}{cc} Y_i& Y_iY_{i-1}\\ Y_iY_{i-1} & Y_iY_{i-1}^2\end{array}\right).
\end{align*}

Table~\ref{table_Poisson} gives the estimated number of change points as well as the detection rates for all three change points in the various setting. In this case, the MOSUM-Wald statistics outperforms the score procedures in terms of detection rates for the second and third change point but indeed the MOSUM-score with the restricted estimator (rather than the global one) slightly outperforms the MOSUM-Wald for the first one. 

The MOSUM-score is computationally much cheaper and scales better with longer series. Indeed, the median computation time (for $100$ runs and $G=n^{2/3}$) for the MOSUM-Wald statistics for a length of $n=1000$ (running several minutes) was more than 272 times slower than the MOSUM-score statistic (running less than a second). For $n=8000$ the MOSUM-Wald statistics ran for more than half an hour which was more than 362 times slower than the MOSUM-score (which ran for less than 6 seconds). In particular, the MOSUM-Wald statistic is significantly slowed down by the need of numerical optimization in comparison to the linear regression where no numerical methods are required.  The numbers only give a qualitative idea where a naive loop was used for the calculation of the statistics and the global and local parameter estimators for the Poisson regression
were calculated with the  \textsf{R}-package \texttt{tscount} \cite{tscountarticle}.

In conclusion, the MOSUM-score statistic can be a good  way to generate change point candidates by means of using different inspection parameters (and bandwidths) even in combination with a sub-optimal covariance estimation procedure.

More simulation results including the false alarm rate in the no-change situation, the results for other bandwidths and covariance estimators, results for the least-squares estimators and more information on computing times  can be found in ~\cite{Reckruehm}, Section 4.2.

\section{Conclusion}\label{section_conclusion}
In this work we propose and analyse a data segmentation procedure in a unified framework including likelihood-based as well as robust methodology for complex non-linear time series such as  Poisson regression or neural-network based (auto-)regression.
To account for the diversity of models included in the framework we work with high-level assumptions in Sections~\ref{section_segmentation} and \ref{section_consistency}, that need to be verified on a case-by-case basis in general. For sufficiently smooth estimating functions this high-level assumptions are shown to be valid under standard moment conditions in Section~\ref{sectionasunderregcond}. Assumptions~\ref{asaltestcovmatwald} for the MOSUM-Wald as well as \ref{as.alt.varianceestimator} for the MOSUM-score procedures and their relaxations in  Remarks~\ref{rem_cov_wald} and \ref{rem_cov_score} on the covariance estimators are very general. Consequently, the choice of covariance estimators has only a minimal influence on the asymptotic properties of the corresponding change point estimators, however, it is critical for the small sample behaviour:  On the one hand (away from change points) the choice of the threshold in combination with the covariance estimator greatly influences the rate of spuriously detected change points. On the other hand (close to change points) this combination determines the detectability rate for a given signal strength. Estimators for the covariance depend crucially on the given model and estimating function, such that they need to be chosen separately for every situation. 

Another contribution of this paper is the analysis of two different procedures: The MOSUM-Wald methodology  yields better results with a single bandwidth but at the cost of a higher computation  time in particular in combination with non-linear models requiring numerical optimization methodology. Furthermore, the MOSUM-Wald  procedure is prone to errors in  a situation where many local maxima are present in the optimization procedure such as is well known for neural network approximations. The reason is that two parameters that are far away in the parameter space can describe almost the same model but will be mistaken for different models by the MOSUM-Wald procedure. The MOSUM-score procedure on the other hand is robust to a multiple mode situations and in addition is computationally much lighter even for non-linear models. On the other hand, with a single inspection parameter it might not be able to detect all change points present (see Section~\ref{sectiondetectabibilty} as well as \ref{section_det_sim}). Future work will aim at solving this problem e.g.\ by combining the methodology with a pruning step as has been discussed for the mean change situation based on the sample mean in \cite{ChoKirch2018}. Such a step is also necessary to deal with truly multiscale heterogeneous situations containing both frequent large jumps as well as small jumps over long stretches of stationarity (for a mathematical definition we refer to \cite{ChoKirch2018}, Definition 2.1).  To deal with such multiscale situations several MOSUM procedures with different bandwidths can be used to generate change point candidates further emphasizing the need for fast candidate generating methods. For the MOSUM-score procedure only the computation of  a variety of data-driven inspection parameters requires numerical optimization while the same inspection parameters  can be used with a variety of bandwidths. The MOSUM-Wald procedure on the other hand requires the (numerical) computation of such estimators for each bandwidth and each point in time making the MOSUM-Wald procedure computationally potentially highly problematic.
For a mathematical analysis of the the pruning step in the mean change situation based on the sample mean, see \cite{ChoKirch2018},  some first considerations in that direction in the framework of this paper can be found in Chapter 5 of \cite{Reckruehm}. 

\section*{Acknowledgements}
This work was supported by the Deutsche Forschungsgemeinschaft (DFG, German Research Foundation) - 314838170, GRK 2297 MathCoRe. Claudia Kirch would also like to thank the Isaac Newton Institute for Mathematical Sciences for support and hospitality during the programme 'Statistical scalability'
when work on this paper was undertaken.

\addcontentsline{toc}{chapter}{References}
\bibliography{literature}

\newpage
\appendix
\section*{Appendix}
\section{Proofs of Section \ref{sectionthreshold}}

\begin{proof}[\textbf{Proof of Theorem \ref{limitdistribution}}]
We first prove the assertions for the MOSUM-score statistics ($\ell=2$):
By the invariance principle in Assumption~\ref{as.invariancemulti} (b) as well as Assumption~\ref{as.bandwidth} on the bandwidth we get (with $\boldsymbol{\Sigma}=\boldsymbol{\Sigma}_1(\widetilde{\theta})$) that
\begin{align*}
&\max_{G\leq k \leq n-G}\frac{1}{\sqrt{2G}}\left\lVert\boldsymbol{\Sigma}^{-1/2}\M_{\tbt}(k)\right\rVert\\
&=\max_{G\leq k \leq n-G}\frac{1}{\sqrt{2G}}\left\lVert\boldsymbol{W}(k+G)-2\boldsymbol{W}(k)+\boldsymbol{W}(k-G)\right\rVert+o_P\left( a(n/G)^{-1} \right)\\
&=\sup_{r \in [G,n-G]}\frac{1}{\sqrt{2G}}\left\lVert\notag
\boldsymbol{W}(r+G)-2\boldsymbol{W}(r)+\boldsymbol{W}(r-G)\right\rVert+o_P\left( a(n/G)^{-1} \right)\\
&\overset{D}{=}
\sup_{t \in [1,n/G -1]}\frac{1}{\sqrt{2}}\left\lVert
\boldsymbol{W}(t+1)-2\boldsymbol{W}(t)+\boldsymbol{W}(t-1)\right\rVert + o_P\left( a(n/G)^{-1} \right),
\end{align*}
where we used the self-similarity of the Wiener process in the last step.
By an application of Lemma 3.1 in combination with Remark 3.1 of \cite{SteinebachEastwood}  with $\alpha=1$ and $C_1=\ldots=C_p=3/2$, assertion (a) for the MOSUM-score statistics follows. The assertion for the MOSUM-Wald statistics ($\ell=1$) follows from this and Assumption~\ref{aswaldaltnullseg}.
The assertions in (b) follow successively by an application of the triangular inequality in combination with the consistency of the spectral matrix norm with the Euclidean vector norm (as the corresponding induced norm).
  The assertions in (c) are obtained analogously.
\end{proof}

\section{Proofs of Section~\ref{chapterwald}}

\begin{proof}[\textbf{Proof of Proposition \ref{propositionsprobminmaxwald}}]
	First, by Theorem~\ref{limitdistribution},
\begin{align*}
&P\left(\max_{j=1,\ldots,q+1}\max_{k_{j-1,n}+G\le k\le k_{j,n}-G} T^{(1)}_{k,n}(G)\ge D_n(\alpha_n,G)\right)\\
&\le \sum_{j=1}^{q+1}P\left(a(n/G) \max_{k_{j-1,n}+G\le k\le k_{j,n}-G} T^{(1)}_{k,n}(G)-b(n/G)\geq c_{\alpha_n} \right)\\
&\le\sum_{j=1}^{q+1}(\alpha_n+o(1))\to 0,
\end{align*}
showing (a) (i). Furthermore, because the spectral matrix norm is induced by the Euclidean vector norm, it holds $\|\mathbf{A}\mathbf{x}\|\le \|\mathbf{A}\| \|\mathbf{x}\|$ as well as  $\|\mathbf{x}\|\le \|\mathbf{A}^{-1}\| 
\|\mathbf{A}\mathbf{x}\|$. Then, by  Assumption~\ref{Ass_31_new} (a) applied to the second term (which is dominated by the maximum in the assumption as can be seen by an index shift),
 \begin{align*}
	&\min_{k_{j,n}-(1-\varepsilon)\,G\le k\le k_{j,n}}
	T^{(1)}_{k,n}(G)		\\
	&\ge \min_{k_{j,n}-(1-\varepsilon)\,G\le k\le k_{j,n}} \sqrt{\frac G2}\left\|\boldsymbol{\Gamma}_k^{-1/2}\left(\boldsymbol{\widehat{\theta}}_{k+1,k+G}-\boldsymbol{\theta}_j\right)\right\|\\*
	&\qquad - \max_{k_{j,n}-(1-\varepsilon)\,G\le k\le k_{j,n}}\sqrt{\frac G2}\left\|\boldsymbol{\Gamma}_k^{-1/2}\left(\boldsymbol{\widehat{\theta}}_{k-G+1,k}-\boldsymbol{\theta}_j\right)\right\|\\
	&\ge \left\|\boldsymbol{\Gamma}_{(j)}^{1/2}\right\|^{-1}\,\min_{k_{j,n}-(1-\varepsilon)\,G\le k\le k_{j,n}} \sqrt{\frac G2}\left\|\boldsymbol{\widehat{\theta}}_{k+1,k+G}-\boldsymbol{\theta}_j\right\|\\*
	&\qquad - \left\|\boldsymbol{\Gamma}_{(j)}^{-1/2}\right\|\,\max_{k_{j,n}-(1-\varepsilon)\,G\le k\le k_{j,n}}\sqrt{\frac G2}\left\|\boldsymbol{\widehat{\theta}}_{k-G+1,k}-\boldsymbol{\theta}_j\right\|\\
	&=\left\|\boldsymbol{\Gamma}_{(j)}^{1/2}\right\|^{-1}\,\min_{k_{j,n}-(1-\varepsilon)\,G\le k\le k_{j,n}} \sqrt{\frac G2}\left\|\boldsymbol{\widehat{\theta}}_{k+1,k+G}-\boldsymbol{\theta}_j\right\|+O_P\left( \sqrt{\log(n/G)} \right).
\end{align*}
We get by Assumption~\ref{as.siglevel} and \ref{Ass_31_new} (b) 
\begin{align*}
	&P\left(\min_{j=1,\ldots,q}\min_{k_{j,n}-(1-\varepsilon)\,G\le k\le k_{j,n}}
	T^{(1)}_{k,n}(G)		< D_n(\alpha_n,G)\right) \\
	&
	\le \sum_{j=1}^{q} P\left(\min_{k_{j,n}-(1-\varepsilon)\,G\le k\le k_{j,n}} \sqrt{\frac G2}\left\|\boldsymbol{\widehat{\theta}}_{k+1,k+G}-\boldsymbol{\theta}_j\right\|< O_P\left(\sqrt{\log(n/G)}\right)  \right)\\*
	&\to 0.
\end{align*}
Similar arguments deal with the minimum over $k_{j,n}<k\le k_{j,n}+(1-\varepsilon)\,G$, such that
assertion (a)(ii) follows.

The proof of (b) follows along the same lines taking Assumption~\ref{asaltestcovmatwald} into account.
\end{proof}
\begin{proof}[\textbf{Proof of Remark \ref{rem_cov_wald}}]
The proof is analogous to the above proofs where we use that by Theorem~\ref{limitdistribution} \begin{align*}
	\max_{j=1,\ldots,q+1}\max_{k_{j-1,n}+G\le k\le k_{j,n}-G}\left\| T^{(1)}_{k,n}(G)\right\|=O_P\left( \sqrt{\log(n/G)} \right).
\end{align*}
\end{proof}

\begin{proof}[\textbf{Proof of Theorem \ref{theorem.consistencyqwald}}]
	The assertion follows immediately from Proposition \ref{propositionsprobminmaxwald} on noting that
	\begin{align*}
		&\left\{\max_{j=1,\ldots,q+1}\max_{k_{j-1,n}+G\le k\le k_{j,n}-G} T^{(1)}_{k,n}(G)< D_n(\alpha_n,G)\right\}\\
		&\quad \cap 
		\left\{
\min_{k_{j,n}-(1-\varepsilon)\,G\le k\le k_{j,n}+(1-\varepsilon)\,G}	T^{(1)}_{k,n}(G)\ge D_n(\alpha_n,G)
		\right\}\\
		&
		\subset\left\{\widehat{q}_n^{(1)}=q\right\}\cap 
\left\lbrace
\max_{1\leq j\leq q}	\left|\widehat{k}^{(1)}_{j,n}-k_{j,n}\right|< G
\right\rbrace.
	\end{align*}
\end{proof}

\section{Proofs of Section~\ref{chapterscore}}

\begin{proof}[\textbf{Proof of Lemma \ref{lemmaproblemdetect}}]
We will show  by contradiction that $\E\bs{H}(\mathbb{X}_1^{(j)},\tbt_{0,1})\neq\E\bs{H}(\mathbb{X}_1^{(j+1)},\tbt_{0,1})$ holds for at least one $j\in\{1,\ldots,q\}$. Assume that all these expectations are equal, then by definition of $\tbt_{0,1}$ we get for all $j=1,\ldots,q+1$
\begin{align*}
 \boldsymbol{0}=\sum_{l=1}^{q+1}\left(\lambda_l-\lambda_{l-1}\right) \E\boldsymbol{H}(\mathbb{X}_1^{(l)},\tbt_{0,1})=\E\boldsymbol{H}(\mathbb{X}_1^{(j)},\tbt_{0,1}),
\end{align*}
which by the identifiability of $\bt_j$ implies $\bt_j=\tbt_{0,1}$, for all $j=1,\ldots,q+1$,
contradicting the assumption. 
If there are only two possible regimes, then clearly if one change is detectable all of them are.
\end{proof}

\begin{proof}[\textbf{Proof of Proposition \ref{propositionsprobminmax}}]
Analogously to the proof of Proposition~\ref{propositionsprobminmaxwald} (a) (i) we get
\begin{align*}
	&P\left(\max_{j=1,\ldots,q+1}\max_{k_{j-1,n}+G\le k\le k_{j,n}-G} T^{(2)}_{k,n}(G,\tbt)\ge D_n(\alpha_n,G)\right)\to 0.
\end{align*}
The statement remains true when a sequence of inspection parameters  $\tbt_n$ is used because by Assumption~\ref{as.replacemulti} (i) it holds for any $j=1,\ldots,q$ that
\begin{align*}
	&a(n/G)\max_{k_{j-1,n}+G\le k\le k_{j,n}-G} T^{(2)}_{k,n}(G,\tbt_n)\\
	&=a(n/G)\max_{k_{j-1,n}+G\le k\le k_{j,n}-G} T^{(2)}_{k,n}(G,\tbt)+o_P(1).
\end{align*}
Additionally, we need the statement for environments of non-detectable change points $k_{j,n}$ with $j\not\in \tilde{Q}$. Indeed, it holds for $|k-k_{j,n}|\le G$ with $j\not\in \tilde Q$ that
$\E \boldsymbol{M}_{\boldsymbol{\widetilde{\theta}}}(k)=0$ by definition of $\tilde Q$ as in \eqref{as.alt.dectest2}.
Consequently, by Assumptions~\ref{as.invariancemulti} and \ref{as.bandwidth}  it holds for $j\not\in\tilde Q$
\begin{align}\label{eq_cons_score_2}
	&\max_{k_{j,n}\le k< k_{j,n}+G}  T^{(2)}_{k,n}(G,\tbt)\notag\\
	&=o_P(1)+O_P\left(\max_{k_{j,n}\le k< k_{j,n}+G}\frac{1}{\sqrt{G}} \left\|\boldsymbol{W}(k+G)-2\boldsymbol{W}(k)+\boldsymbol{W}(k_{j,n})\right\| \right)\notag\\
	&\qquad+O_P\left(\max_{k_{j,n}\le k< k_{j,n}+G}\frac{1}{\sqrt{G}}\left\|\boldsymbol{\Sigma}^{-1/2}_{(j+1)}\boldsymbol{\Sigma}^{1/2}_{(j)}\left(\boldsymbol{W}(k_{j,n})-\boldsymbol{W}(k-G))\right)\right\|  \right)\notag\\
	&=O_P(1)=o_P(D_n(\alpha_n,G)),
\end{align}
where the last line follows by the self-similarity of  Wiener processes, the stationarity of its increments and the continuous sample paths. An analogous assertion holds for $k_{j,n}-G\le k< k_{j,n}$ showing that
\begin{align*}
	&P\left(\max_{j\not\in \tilde Q} \max_{|k-k_{j,n}|<G} T^{(2)}_{k,n}(G,\tbt)\ge D_n(\alpha_n,G)\right)\to 0,
\end{align*}
completing the proof of (a) (i) for a fixed inspection parameter $\tbt$. Here, the statement remains true for a sequence of inspection parameters, because by Assumption~\ref{as.replacemulti} (ii) it holds for  $j\not\in \tilde{Q}$
\begin{align}\label{eq_cons_score}
	&\max_{|k-k_{j,n}|<G} T^{(2)}_{k,n}(G,\tbt_n)
	=\max_{|k-k_{j,n}|<G} T^{(2)}_{k,n}(G,\tbt)+o_P(\sqrt{\log(n/G)})\\
	&=o_P(D_n(\alpha_n,G))\notag.\end{align}

	The assertion with estimated long-run covariances as in (b) can be obtained along the same lines by using the consistency of the spectral matrix norm with the Euclidean vector norm and Assumptions~\ref{as.alt.varianceestimator} (b). 

Concerning (ii) first observe that for $\tilde k_{j,n}<k\le \tilde k_{j,n}+(1-\varepsilon)G$, $j=1,\ldots,\tilde q(\tbt)$ it holds
\begin{align*}
	&\E\M_{\tbt}(k)\\
&=\sum_{i=k+1}^{k+G}\E\boldsymbol{H}(\mathbb{X}_{i}^{(j+1)},\tbt)- \sum_{i=k-G+1}^{\tilde k_{j,n}}\E\boldsymbol{H}(\mathbb{X}_{i}^{(j)},\tbt)-\sum_{i=\tilde k_{j,n}+1}^{k}\E\boldsymbol{H}(\mathbb{X}_{i}^{(j+1)},\tbt)\\
&=\left(G-|k-\tilde k_{j,n}|\right)\,\boldsymbol{d}_j,
	\end{align*}
	where we denote the signal by
	\begin{align}
		&\boldsymbol{d}_j=\E\boldsymbol{H}(\mathbb{X}_{1}^{(j+1)},\tbt)-\E\boldsymbol{H}(\mathbb{X}_{1}^{(j)},\tbt),\quad j=1,\ldots\tilde q(\tbt).
		\label{eq_signal_score}
	\end{align}
	For $\tilde k_{j,n}-(1-\varepsilon)G\le k \le \tilde k_{j,n}$ we arrive at the same conclusion. Consequently, it holds  for all $j\in\tilde Q$ and $|k-\tilde k_{j,n}|\le (1-\varepsilon) G$ by the consistency of the spectral matrix with the Euclidean vector norm 
	\begin{align*}
		&&\left\lVert\boldsymbol{\Sigma}_k^{-1/2}\E\M_{\tbt}(k)\right\rVert\ge \left\|\boldsymbol{\Sigma}_k^{1/2}\right\|^{-1}\,\left(G-|k-\tilde k_{j,n}|\right)\,\left\lVert\boldsymbol{d}_j\right\rVert\geq c \,G,	\end{align*}
for some $c>0$ (depending on $\varepsilon$, the difference in expectation and the long-run covariances, noting that $\boldsymbol{\Sigma}_k$ is constant on each segment).

By analogous arguments as in \eqref{eq_cons_score_2} (but involving the necessary centering due to  $\tilde Q$) and \eqref{eq_cons_score}  it holds 
\begin{align*}
	&\min_{|k-\tilde k_{j,n}|\le (1-\varepsilon)\,G} T^{(2)}_{k,n}(G,\tbt_n)	\\
	&= \min_{|k-\tilde k_{j,n}|\le (1-\varepsilon)\, G}\frac{1}{\sqrt{2G}}\,\left\lVert\boldsymbol{\Sigma}_k^{-1/2}\E\M_{\tbt}(k)\right\rVert+ o_P(\sqrt{\log(n/G)})\\
	&\ge c\, \sqrt{\frac G 2}+o_P\left( \sqrt{\log(n/G)} \right).
\end{align*}
The proof can now be concluded as in the proof of  Proposition~\ref{propositionsprobminmaxwald} (a) and (b) (ii).
\end{proof}

\begin{proof}[\textbf{Proofs of Theorem \ref{theorem.consistencyq} and Remark~\ref{rem_cov_score}}]
	The proofs are completely analogous to the proofs of Theorem~\ref{theorem.consistencyqwald} respectively Remark~\ref{rem_cov_wald} and therefore omitted.
\end{proof}

The following lemma helps simplify several arguments.
\begin{lemma}\label{lemma_bound_conv}
	For a sequence of real random variables $\{X_n\}$ it holds for $n\to\infty$
	\begin{align*}
		& X_n=O_P(1) \quad \iff\quad P\left( |X_n|>\xi_n \right)\to 0\quad \text{for any }\xi_n\to \infty.
	\end{align*}
\end{lemma}

\begin{proof}
The proof of the only-if-part is straightforward. We prove the if-part by contradiction. If $X_n$ is not stochastically bounded, then there exists $\eta>0$ such that for any bound $C>0$ and any $n_0\ge 0$, there exists $n_1(n_0,C)>n_0$ such that 
	$P\left( |X_{n_1}|> C\right)>\eta.$
Setting $N_0=0$ and recursively $N_l=n_1(N_{l-1},l)$ as well as $\xi_{N_{l-1}+1}=\ldots=\xi_{N_l}=l$, we get $N_n\to\infty$, $\xi_n\to\infty$ as well as by construction
\begin{align*}
	P\left( |X_{N_l}|>\xi_{N_l} \right)>\eta,
\end{align*}
which is a contradiction.
\end{proof}

The proof technique of the below proof is well known in change point analysis, for example it has been used in the context of MOSUM statistics for the mean change problem by \cite{KirchMuhsal} (Proof of Theorem 3.2).

\begin{proof}[\textbf{Proof of Theorem \ref{theoremimprovedconvrates}}]
	By finiteness of $q$ and Lemma \ref{lemma_bound_conv} it is sufficient to prove that for any sequence $\xi_n\to\infty$ (arbitrarily slow) it holds
\begin{align*}
	&P\left( \widehat{k}^{(2)}_{j,n}(\tbt_n; \widehat{\boldsymbol{\Psi}}_{j,n})<\tilde k_{j,n}(\tbt)-\xi_n \right)\to 0,\\ \quad 
&	P\left( \widehat{k}^{(2)}_{j,n}(\tbt_n; \widehat{\boldsymbol{\Psi}}_{j,n})>\tilde k_{j,n}(\tbt)+\xi_n \right)\to 0.
\end{align*}
We will prove the first assertion in detail, the second one follows analogously. For simplicity of notation denote $\tilde k_{j,n}=\tilde{k}_{j,n}(\tbt)$ throughout this proof. 

On the asymptotic 1-set of Theorem~\ref{theorem.consistencyq} and where $\min_{j=1,\ldots,q+1}|k_{j,n}-k_{j-1,n}|>2G$ (which holds for $n$ large enough by Assumption~\ref{as.bandwidth} (b)), it holds for any $j=1,\ldots,\tilde{q}(\tbt)$ with $\boldsymbol{d}_j$ as in \eqref{eq_signal_score}
\begin{align*}
	&\widehat k_{j,n}^{(2)}(\tbt_n; \widehat{\boldsymbol{\Psi}}_{j,n})
	=\argmax_{v_{j,n}\le k\le w_{j,n}}V^{(j)}_{k,n}(G,\tbt_n),\quad \text{where }\\
	&V^{(j)}_{k,n}(G,\tbt_n)=\left\lVert\boldsymbol{\widehat{\Psi}}_{j,n}^{-1/2}\M_{\tbt_n}(k)\right\rVert^2-\left\lVert\boldsymbol{\widehat{\Psi}}_{j,n}^{-1/2}\M_{\tbt_n}(k_{j,n})\right\rVert^2\\
	&=-\left(\M_{\tbt_n}(\tilde k_{j,n})-\M_{\tbt_n}(k)  \right)\,\boldsymbol{\widehat{\Psi}}_{j,n}^{-1}\,\left( \M_{\tbt_n}(\tilde k_{j,n})+\M_{\tbt_n}(k) \right)\\
&	=:-\left( \boldsymbol{E}_1(k,G,\tbt_n) +\boldsymbol{d}_j (\tilde k_{j,n}-k)\right)\,\boldsymbol{\widehat{\Psi}}_{j,n}^{-1}\,\left( \boldsymbol{E}_2(k,G,\tbt_n) +\boldsymbol{d}_j (2G+k-\tilde k_{j,n})\right).
\end{align*}
Denote  $\Delta\boldsymbol{H}(\mathbb{X}_{i}^{(j)},\tbt_n)=\boldsymbol{H}(\mathbb{X}_{i}^{(j)},\tbt_n)-\boldsymbol{H}(\mathbb{X}_{i}^{(j)},\tbt)$ and $\boldsymbol{H}_0(\mathbb{X}_{i}^{(j)},\tbt)=\boldsymbol{H}(\mathbb{X}_{i}^{(j)},\tbt)-\E \boldsymbol{H}(\mathbb{X}_{i}^{(j)},\tbt)$. Then, it
holds for $k<\tilde{k}_{j,n}$ \begin{align*}
	& \boldsymbol{E}_1(k,G,\tbt_n)=\M_{\tbt_n}(\tilde k_{j,n})-\M_{\tbt_n}(k)-\boldsymbol{d}_j (\tilde k_{j,n}-k)\\
	&=\sum_{i=k+G+1}^{\tilde k_{j,n}+G}\Delta\boldsymbol{H}(\mathbb{X}_{i}^{(j+1)},\tbt_n)+\sum_{i=k-G+1}^{\tilde k_{j,n}-G}\Delta\boldsymbol{H}(\mathbb{X}_{i}^{(j)},\tbt_n)-2\sum_{i=k+1}^{\tilde k_{j,n}}\Delta\boldsymbol{H}(\mathbb{X}_{i}^{(j)},\tbt_n)
	\\*
	&\qquad+\sum_{i=k+G+1}^{\tilde k_{j,n}+G}\boldsymbol{H}_0(\mathbb{X}_{i}^{(j+1)},\tbt)+\sum_{i=k-G+1}^{\tilde k_{j,n}-G}\boldsymbol{H}_0(\mathbb{X}_{i}^{(j)},\tbt)-2\sum_{i=k+1}^{\tilde k_{j,n}}\boldsymbol{H}_0(\mathbb{X}_{i}^{(j)},\tbt),
\end{align*}
where by Assumption \ref{ass_localization_rates} (see also Remark~\ref{rem_localization_rates} for the situation when $k>\tilde k_{j,n}$) and stationarity of the segments it follows
\begin{align}\label{eq_E1}
	&\sup_{\xi_n< \tilde k_{j,n}-k\le G}\frac{\|\boldsymbol{E}_1(k,G,\tbt_n)\|}{\tilde k_{j,n}-k}=o_P(1).
\end{align}
Furthermore,
\begin{align*}
	&\boldsymbol{E}_2(k,G,\tbt_n)\\
	&=-\boldsymbol{E}_1(k,G,\tbt_n)\\
	&\quad +2 \sum_{i=\tilde k_{j,n}+1}^{\tilde k_{j,n}+G}\Delta\boldsymbol{H}(\mathbb{X}_{i}^{(j+1)},\tbt_n)-2\sum_{i=\tilde k_{j,n}-G+1}^{\tilde k_{j,n}}\Delta\boldsymbol{H}(\mathbb{X}_{i}^{(j)},\tbt_n)
\\&\quad +2 \sum_{i=\tilde k_{j,n}+1}^{\tilde k_{j,n}+G}\boldsymbol{H}_0(\mathbb{X}_{i}^{(j+1)},\tbt)-2\sum_{i=\tilde k_{j,n}-G+1}^{\tilde k_{j,n}}\boldsymbol{H}_0(\mathbb{X}_{i}^{(j)},\tbt),\end{align*}
such that by \eqref{eq_E1}, Assumption~\ref{ass_localization_rates} (a) and stationarity of the segments in combination with the law of large number (that follows from Assumption~\ref{as.invariancemulti}, see also Remark~\ref{rem_localization_rates}) it holds
\begin{align*}
	&\sup_{\xi_n<\tilde k_{j,n}-k\le G}\frac{\|\boldsymbol{E}_2(k,G,\tbt_n)\|}{G}=o_P(1).
\end{align*}
By Assumptions~\ref{ass_Psi} and the consistency of the spectral matrix norm with the Euclidean vector norm we conclude 
\begin{align*}
	&\sup_{\xi_n< \tilde k_{j,n}-k\le G}V^{(j)}_{k,n}(G,\tbt_n)\\
	&\le (\boldsymbol{d}_j\boldsymbol{\widehat{\Psi}}_{j,n}^{-1}\boldsymbol{d}_j+o_P(1))\sup_{\xi_n<\tilde k_{j,n}-k\le G}[-(\tilde k_{j,n}-k)(2G+k-\tilde k_{j,n})].
\end{align*}
Finally,
\begin{align*}
\sup_{\xi_n< \tilde k_{j,n}-k\le G}[-(\tilde k_{j,n}-k)(2G+k-\tilde k_{j,n})]\le -\xi_n G<0,
\end{align*}
and
\begin{align*}
	&\boldsymbol{d}_j\boldsymbol{\widehat{\Psi}}_{j,n}^{-1}\boldsymbol{d}_j=\|\boldsymbol{\widehat{\Psi}}_{j,n}^{-1/2}\boldsymbol{d}_j\|^2\ge \|\boldsymbol{\widehat{\Psi}}_{j,n}^{1/2}\|^{-2}
	\|\boldsymbol{d}_j\|^2\end{align*}
	Consequently,
	\begin{align*}
		&P\left(  \widehat{k}^{(2)}_{j,n}(\tbt_n; \widehat{\boldsymbol{\Psi}}_{j,n})<\tilde k_{j,n}-\xi_n \right)\\
		&\le P\left(\sup_{\tilde k_{j,n}-G\le k< \tilde k_{j,n}-\xi_n}V^{(j)}_{k,n}(G,\tbt_n)\ge \sup_{\tilde  k_{j,n}-\xi_n\le k\le \tilde  k_{j,n}+G}V^{(j)}_{k,n}(G,\tbt_n)  \right)+o(1)\\
		&\le P\left( \sup_{\tilde k_{j,n}-G\le k< \tilde k_{j,n}-\xi_n}V^{(j)}_{k,n}(G,\tbt_n)\ge 0 \right)+o(1)\\
		&\le P\left( (\boldsymbol{d}_j\boldsymbol{\widehat{\Psi}}_{j,n}^{-1}\boldsymbol{d}_j+o_P(1))\sup_{\xi_n<\tilde k_{j,n}-k\le G}[-(\tilde k_{j,n}-k)(2G+k-\tilde k_{j,n})]\ge 0 \right)+o(1)\\
		&\le P\left( |o_P(1)|\ge \boldsymbol{d}_j\boldsymbol{\widehat{\Psi}}_{j,n}^{-1}\boldsymbol{d}_j \right)+o(1)\le P\left( |o_P(1)|\,\|\boldsymbol{\widehat{\Psi}}_{j,n}^{1/2}\|^2\ge \|\boldsymbol{d}_j\|^2 \right)+o(1)\\
		&=o(1),
	\end{align*}
	where the last line follows from   Assumption~\ref{ass_Psi}, concluding the proof.
\end{proof}

\section{Proofs of Section~\ref{sectionasunderregcond}}

\begin{proof}[Proof of Theorem~\ref{theorem_ass_ip}]
	The assertion in (a) follows by  Theorem 4 of \cite{kuelbsphilipp80} on noting that the mixing rate of $\boldsymbol{H}(\mathbb{X}_1^{(j)},\boldsymbol{\theta})$ is at least as good as the one of $\{\mathbb{X}_1^{(j)}\}$ by definition. Because the time series in backward time is also mixing with the same rate, Assumption~\ref{ass_localization_rates} (b) follow from the invariance principle in backward time (see also Remark~\ref{rem_localization_rates}).
\end{proof}

The remaining assumptions all correspond to well known results in statistics if a global estimator based on estimating functions is used. However, here, it is  maximized over an increasing number of windows. To this end, we require versions of uniform laws of large numbers taking these moving windows into account as given in the following lemma. 
\begin{lemma}\label{lemma_uniform}Let $\{\mathbb{Y}_t\}$ be $p$-dimensional random vectors  fulfilling Regularity Condition~\ref{regc_ip} (b) with $\tilde{\nu}$ as below (in (a),(b),(d) and arbitrary in (c)),   $\boldsymbol{\Theta}\subset\mathbb{R}^p$ be a compact parameter space and $\boldsymbol{F}=(F_1,\ldots,F_p)^T:(\mathbb{R}^p,\boldsymbol{\Theta})\to\mathbb{R}^p$ measurable.
	\begin{enumerate}[(a)]
		\item If $0<\E \left\|\boldsymbol{F}(\mathbb{Y}_1,\boldsymbol{\theta})\right\|^{2+\tilde{\nu}}<\infty$ for some $\tilde{\nu}>0$ and some $\boldsymbol{\theta}$, then for the same $\boldsymbol{\theta}$
			\begin{align*}
				&\sup_{0\le k\le n-G}\left\|\sum_{i=k+1}^{k+G}\left(\boldsymbol{F}(\mathbb{Y}_i,\boldsymbol{\theta})-\E\left( \boldsymbol{F}(\mathbb{Y}_1,\boldsymbol{\theta})\right) \right)\right\|=O_P\left( \sqrt{G\,\log(n/G)} \right),\\
						\end{align*}
		\item If for some $\tilde{\nu}>0$ it holds  $0<\E \left\|\boldsymbol{F}(\mathbb{Y}_1,\boldsymbol{\theta})\right\|^{2+\tilde{\nu}}<\infty$ for all $\boldsymbol{\theta}$ as well as\\
		$		\E\sup_{\boldsymbol{\theta}\in\boldsymbol{\Theta}}\| \nabla \boldsymbol{F}(\mathbb{Y}_1,\boldsymbol{\theta})\|^{2+\tilde\nu}<\infty,$
			then
			\begin{align*}
			\sup_{\boldsymbol{\theta}\in \boldsymbol{\Theta}}\max_{0\leq k\leq n-G}\frac{1}{G}\left\lVert\sum_{i=k+1}^{k+G}( \boldsymbol{F}(\mathbb{Y}_i,\boldsymbol{\theta})-\E\left(\boldsymbol{F}(\mathbb{Y}_1,\boldsymbol{\theta})\right))\right\rVert=o_P(1).
		\end{align*}
	\item  If $\E\sup_{\boldsymbol{\theta}\in\boldsymbol{\Theta}}\| \boldsymbol{F}(\mathbb{Y}_1,\boldsymbol{\theta})\|<\infty$, then for any sequence $G\to\infty$	it holds
		\begin{align*}
			(i)\quad &\sup_{\boldsymbol{\theta}\in \boldsymbol{\Theta}}\max_{1\leq k\leq G}\frac{1}{k}\left\lVert\sum_{i=G-k+1}^{G} \boldsymbol{F}(\mathbb{Y}_i,\boldsymbol{\theta})\right\rVert=O_P(1),\quad\\
			&\sup_{\boldsymbol{\theta}\in \boldsymbol{\Theta}}\max_{1\leq k\leq G}\frac{1}{k}\left\lVert\sum_{i=1}^{k} \boldsymbol{F}(\mathbb{Y}_i,\boldsymbol{\theta})\right\rVert=O_P(1).\\
			(ii)\quad&\sup_{\boldsymbol{\theta}\in \boldsymbol{\Theta}}\max_{1\leq k\leq G}\frac{1}{G}\left\lVert\sum_{i=1}^{k}( \boldsymbol{F}(\mathbb{Y}_i,\boldsymbol{\theta})-\E\left(\boldsymbol{F}(\mathbb{Y}_1,\boldsymbol{\theta})\right))\right\rVert=o_P(1),\\
&\sup_{\boldsymbol{\theta}\in \boldsymbol{\Theta}}\max_{1\leq k\leq G}\frac{1}{G}\left\lVert\sum_{i=k}^{G}( \boldsymbol{F}(\mathbb{Y}_i,\boldsymbol{\theta})-\E\left(\boldsymbol{F}(\mathbb{Y}_1,\boldsymbol{\theta})\right))\right\rVert=o_P(1).
		\end{align*}
	\item If $\E\sup_{\boldsymbol{\theta}\in\boldsymbol{\Theta}}\| \boldsymbol{F}(\mathbb{Y}_1,\boldsymbol{\theta})\|^{2+\tilde{\nu}}<\infty$, then
		\begin{align*}
			&\sup_{0\le k\le n-G}\sum_{i=k+1}^{k+G}\sup_{\boldsymbol{\theta}\in\boldsymbol{\Theta}}\left\|\boldsymbol{F}(\mathbb{Y}_i,\boldsymbol{\theta})\right\|=O_P\left( G \right).
		\end{align*}
\end{enumerate}
\end{lemma}
\begin{proof}
	Analogously to Theorem~\ref{theorem_ass_ip} $\{\boldsymbol{F}(\mathbb{Y}_i,\boldsymbol{\theta})\}$ fulfills an invariance principle from which assertion (a) follows by similar arguments as in the proof of Theorem~\ref{limitdistribution} (see \cite{Reckruehm}, Theorem E.2.12 for details).

	The proof technique for (b) is well known (and we only use a basic version thereof). Thus, we only sketch the proof. First note that by the compactness assumption on $\boldsymbol{\Theta}$ for  each $\delta>0$ there exist  $M=M(\delta)\ge 1$ and $\boldsymbol{\xi}_1,\ldots,\boldsymbol{\xi}_M\in\boldsymbol{\Theta}$ such that for any $\boldsymbol{\theta}\in\boldsymbol{\Theta}$ there is an $m=1,\ldots, M$ with $\left\lVert\boldsymbol{\theta}-\boldsymbol{\xi}_m\right\rVert<\delta$. 
	We get for any $\boldsymbol{\xi},\bt$
\begin{align*}
	&\max_{0\le k\le n-G}\frac 1 G \sum_{i=k+1}^{k+G}\left\lVert\boldsymbol{F}(\mathbb{Y}_i,\boldsymbol{\theta})-\E\boldsymbol{F}(\mathbb{Y}_i,\boldsymbol{\theta}) -\left(\boldsymbol{F}(\mathbb{Y}_i,\boldsymbol{\xi})-\E\boldsymbol{F}(\mathbb{Y}_i,\boldsymbol{\xi})\right)\right\rVert\\
	&\le\left(2 \E\sup_{\bt\in\boldsymbol{\Theta}}\|\nabla\boldsymbol{F}(\mathbb{Y}_i,\boldsymbol{\theta})\|+o_P(1)\right) \, \|\boldsymbol{\theta}-\boldsymbol{\xi}\|
	=\|\boldsymbol{\theta}-\boldsymbol{\xi}\|\,O_P(1),
\end{align*}
where the last line follows from a first order Taylor expansion in addition to a moving law of large numbers as in (a) applied to the time series $\{\sup_{\bt\in\boldsymbol{\Theta}}\|\nabla\boldsymbol{F}(\mathbb{Y}_i,\boldsymbol{\theta})\|\}$ (see also (d)). 
For given $\eta_1,\eta_2>0$ we can now choose $\delta=\delta(\eta_1,\eta_2)>0$ such that
\begin{align*}
	&P\Big( \sup_{\|\bt-\boldsymbol{\xi}\|<\delta}\max_{0\le k\le n-G}\frac 1 G \sum_{i=k+1}^{k+G}\left\lVert\boldsymbol{F}(\mathbb{Y}_i,\boldsymbol{\theta})-\E\boldsymbol{F}(\mathbb{Y}_i,\boldsymbol{\theta})\right.\\
	&\phantom{P( \max_{\|\bt-\boldsymbol{\xi}\|<\delta}\max_{0\le k\le n-G}\frac 1 G \sum_{i=k+1}^{k+G}}\quad
	\left.-\left(\boldsymbol{F}(\mathbb{Y}_i,\boldsymbol{\xi})-\E\boldsymbol{F}(\mathbb{Y}_i,\boldsymbol{\xi})\right)\right\rVert\ge \eta_1 \Big)\le \eta_2
\end{align*}
for all $n\ge n_0(\eta_1,\eta_2)$. For the (to $\delta$) corresponding $\boldsymbol{\xi_1},\ldots,\boldsymbol{\xi_M}$ it holds by another
 application of (a) 
 \begin{align*}
		&P\left( \max_{m=1,\ldots,M} \max_{0\leq k\leq n-G}\frac{1}{G}\left\lVert\sum_{i=k+1}^{k+G}( \boldsymbol{F}(\mathbb{Y}_i,\boldsymbol{\xi}_m)-\E\left(\boldsymbol{F}(\mathbb{Y}_1,\boldsymbol{\xi}_m)\right))\right\rVert\ge \eta_1 \right)\le \eta_2.
	\end{align*}
	for all $n\ge n_1(\eta_1,\eta_2)$. Combining these arguments yields (b) on noting that for any $\delta$ and corresponding $\boldsymbol{\xi}_m$, $m=1,\ldots,M$, it holds
	\begin{align*}
		&\sup_{\boldsymbol{\theta}\in \boldsymbol{\Theta}}\max_{0\leq k\leq n-G}\frac{1}{G}\left\lVert\sum_{i=k+1}^{k+G}( \boldsymbol{F}(\mathbb{Y}_i,\boldsymbol{\theta})-\E\left(\boldsymbol{F}(\mathbb{Y}_1,\boldsymbol{\theta})\right))\right\rVert\\
		&\le \max_{m=1,\ldots,M} \max_{0\leq k\leq n-G}\frac{1}{G}\left\lVert\sum_{i=k+1}^{k+G}( \boldsymbol{F}(\mathbb{Y}_i,\boldsymbol{\xi}_m)-\E\left(\boldsymbol{F}(\mathbb{Y}_1,\boldsymbol{\xi}_m)\right))\right\rVert\\
		&\qquad + \sup_{\|\bt-\boldsymbol{\xi}\|<\delta}\max_{0\le k\le n-G}\frac 1 G \sum_{i=k+1}^{k+G}\left\lVert\boldsymbol{F}(\mathbb{Y}_i,\boldsymbol{\theta})-\E\boldsymbol{F}(\mathbb{Y}_i,\boldsymbol{\theta})\right.\\
	&\phantom{\qquad + \sup_{\|\bt-\boldsymbol{\xi}\|<\delta}\max_{0\le k\le n-G}\frac 1 G \sum_{i=k+1}^{k+G}\quad
	}
	\left.-\left(\boldsymbol{F}(\mathbb{Y}_i,\boldsymbol{\xi})-\E\boldsymbol{F}(\mathbb{Y}_i,\boldsymbol{\xi})\right)\right\rVert.
	\end{align*}

By \cite{Rao}, Theorem 6.5, a uniform (in $\boldsymbol{\theta}$) strong law of large numbers holds for $\{\boldsymbol{F}(\mathbb{Y}_i,\boldsymbol{\theta})\}$ because stationarity and mixing implies ergodicity (both forward and backward). By the almost sure convergence standard arguments give the assertions in (c).

The proof of (d) follows along the same lines as the proof of (a) but applied to the function $\sup_{\boldsymbol{\theta}\in\boldsymbol{\Theta}}\left\|\boldsymbol{F}(\mathbb{Y}_i,\boldsymbol{\theta})\right\|$. The necessary centering is of the order $G$.
\end{proof}

\begin{proof}[Proof of Theorem~\ref{theorem_ass_wald}]
	For $k_{j-1,n}<k\le k_{j,n}-G $ a Taylor expansion in $\boldsymbol{\widehat{\theta}}_{k+1,k+G}$  around $\boldsymbol{\theta}_j$  yields that there exists  $\left\lVert \boldsymbol{\xi}^{(j)}_{k,n}-\boldsymbol{\theta}_j\right\rVert\leq \left\lVert \boldsymbol{\widehat{\theta}}_{k+1,k+G}-\boldsymbol{\theta}_j\right\rVert$ such that
\begin{align*}
	&-\frac{1}{\sqrt{G}}\sum_{i=k+1}^{k+G}\boldsymbol{H}(\mathbb{X}_i^{(j)},\boldsymbol{\theta_j})\\
	&=\left(\frac{1}{G}\sum_{i=k+1}^{k+G}\nabla\boldsymbol{H}(\mathbb{X}^{(j)}_i,\boldsymbol{\xi}_{k,n}^{(j)})\right)^T\sqrt{G}\left(\boldsymbol{\widehat{\theta}}_{k+1,k+G}-\boldsymbol{\theta}_j\right)\\
	&=\left(o_P(1)+\bs{V}_{(j)}(\bs{\xi}_{k,n}^{(j)})\right) \sqrt{G}\left(\boldsymbol{\widehat{\theta}}_{k+1,k+G}-\boldsymbol{\theta}_j\right)\text{ uniformly in } k,
\end{align*}
where the last line follows by Regularity Conditions~\ref{regc_wald} (a) in combination with Lemma~\ref{lemma_uniform} (b). By  Lemma~\ref{lemma_uniform} (a), Regularity Condition~\ref{regc_ip} and the definition of $\boldsymbol{\theta}_j$ we get 
\begin{align*}	&\sup_{1\le k\le n-G}\frac{1}{\sqrt{G}}\left\|\sum_{i=k+1}^{k+G}\boldsymbol{H}(\mathbb{X}_i^{(j)},\boldsymbol{\theta_j})\right\|=O_P\left( \sqrt{\log(n/G)} \right).
\end{align*}
In combination with  Regularity Conditions~\ref{regc_wald}(b) this yields
\begin{align}\label{eq_315}
	&\max_{j=1,\ldots,q+1}\max_{k_{j-1,n}< k \leq k_{j,n}-G}\sqrt{G}\left\lVert\boldsymbol{\widehat{\theta}}_{k+1,k+G}-\boldsymbol{\theta}_j\right\rVert=O_P\left(\sqrt{\log(n/G)}\right).
\end{align}
In particular, this shows the validity of Assumption~\ref{Ass_31_new} (a).\\
Moreover, for each $l=1,\ldots,p$ and $k_{j-1,n}<k\le k_{j,n}-G $ a second order Taylor expansion  yields
the existence of 
$ \left\lVert \boldsymbol{\xi}^{(j)}_{l,n,k}-\boldsymbol{\theta}_j\right\rVert\leq \left\lVert \widehat{\boldsymbol{\theta}}_{k+1,k+G}-\boldsymbol{\theta}_j\right\rVert$ with
\begin{align*}
	&-\sum_{i=k+1}^{k+G}H_l(\mathbb{X}_i^{(j)},\boldsymbol{\theta}_j)\\
	&=\left(\sum_{i=k+1}^{k+G}\left(\nabla H_l(\mathbb{X}_i^{(j)},\boldsymbol{\theta}_j)\right)\right)^T\left(\boldsymbol{\widehat{\theta}}_{k+1,k+G}-\boldsymbol{\theta}_j\right)\\*
	&\quad+\frac{1}{2}\left(\boldsymbol{\widehat{\theta}}_{k+1,k+G}-\boldsymbol{\theta}_j\right)^T\left(\sum_{i=k+1}^{k+G}\nabla^2 H_l(\mathbb{X}^{(j)}_i,\boldsymbol{\xi}^{(j)}_{l,n,k})\right)\left(\boldsymbol{\widehat{\theta}}_{k+1,k+G}-\boldsymbol{\theta}_j\right)\\
	&=\left(\E\left(\nabla H_l(\mathbb{X}_1^{(j)},\boldsymbol{\theta}_j)\right)+O_P\left( \sqrt{\frac{\log (n/G)}{G}} \right)  \right)^T \; G \left(\boldsymbol{\widehat{\theta}}_{k+1,k+G}-\boldsymbol{\theta}_j\right)\\
	&\quad + O_P\left( G \right) \|\boldsymbol{\widehat{\theta}}_{k+1,k+G}-\boldsymbol{\theta}_j\|^2\quad \text{uniformly in }k,
\end{align*}
where the last line follows by an application of Lemma~\ref{lemma_uniform} both (a) and (d).
Thus, an application of \eqref{eq_315} yields uniformly in $k$
\begin{align*}
	&-\frac{1}{\sqrt{2G}}\sum_{i=k+1}^{k+G}H_l(\mathbb{X}_i^{(j)},\boldsymbol{\theta}_j)\\&=\E\left(\nabla H_l(\mathbb{X}_1^{(j)},\boldsymbol{\theta}_j)\right)^T\sqrt{\frac G 2}\left(\boldsymbol{\widehat{\theta}}_{k+1,k+G}-\boldsymbol{\theta}_j\right)+ o_P\left( \left(\log n/G\right)^{-1/2} \right), 
\end{align*}
showing the validity of Assumption~\ref{aswaldaltnullseg}, concluding the proof of (a). By the strong law of large numbers (and similar arguments as in the proof of Lemma~\ref{lemma_uniform} (c)) and 
by definition of $\boldsymbol{\theta}_j$ it holds
\begin{align*}
	&\max_{	k_{j-1,n}-G\le k\le k_{j-1,n}-\varepsilon\,G }\left\|\frac 1 G\sum_{i=k+1}^{k+G}\boldsymbol{H}(\mathbb{X}_i,\boldsymbol{\theta}_j)-\frac{k_{j-1,n}-k}{G}
	\E \boldsymbol{H}(\mathbb{X}_i^{(j-1)},\boldsymbol{\theta}_j)\right\|\\
& =\max_{	k_{j-1,n}-G\le k\le k_{j-1,n}-\varepsilon\,G }
\left\|\frac 1 G\sum_{i=k+1}^{k_{j-1,n}}\left(\boldsymbol{H}(\mathbb{X}_i^{(j-1)},\boldsymbol{\theta}_j)-\E \boldsymbol{H}(\mathbb{X}_i^{(j-1)},\boldsymbol{\theta}_j)\right)\right.\\*
& 
\phantom{ =\max_{	k_{j-1,n}-G\le k\le k_{j-1,n}-\varepsilon\,G }}\quad
+ \left. 
\frac 1 G\sum_{i=k_{j-1,n}+1}^{k+G}\boldsymbol{H}(\mathbb{X}_i^{(j)},\boldsymbol{\theta}_j)
\right\|=o_P(1).
\end{align*} By the identifiable uniqueness of $\boldsymbol{\theta}_j$  it holds $\E \boldsymbol{H}(\mathbb{X}_i^{(j-1)},\boldsymbol{\theta}_j)\neq 0$, such that \begin{align*}
	&\sqrt{\frac{G}{\log(n/G)}}\,\min_{k_{j-1,n}-G\le k\le k_{j-1,n}-\varepsilon\,G}
	\left\|\frac 1 G\sum_{i=k+1}^{k+G}\boldsymbol{H}(\mathbb{X}_i,\boldsymbol{\theta}_j)\right\|\\
	&\ge \sqrt{\frac{G}{\log(n/G)}}\left(    \varepsilon\, \left\|\E \boldsymbol{H}(\mathbb{X}_i^{(j-1)},\boldsymbol{\theta}_j)\right\| +o_P(1)\right)\overset{P}{\longrightarrow} \infty.
\end{align*} 
By Lemma~\ref{lemma_uniform} (c) it holds 
\begin{align*}
	&\sup_{\bt\in\Theta}\max_{k_{j-1,n}-G\le k\le k_{j-1,n}-\varepsilon\,G
	}\left\|\frac{1}{G}\sum_{i=k+1}^{k+G}\nabla\boldsymbol{H}(\mathbb{X}_i,	\bt)\right\|_F=O_P(1),
\end{align*}
where $\|\cdot\|_F$ denotes the Frobenius matrix norm. 

Furthermore, a Taylor expansion of $\boldsymbol{\widehat{\theta}}_{k+1,k+G}$ around $\boldsymbol{\theta}_j$ yields for some $\boldsymbol{\xi}_{k,n}$\begin{align*}
	&\sqrt{\frac{G}{\log(n/G)}}\,	\left\|\frac 1 G\sum_{i=k+1}^{k+G}\boldsymbol{H}(\mathbb{X}_i,\boldsymbol{\theta}_j)\right\|\\
&=\left\|\left(\frac{1}{G}\sum_{i=k+1}^{k+G}\nabla\boldsymbol{H}(\mathbb{X}_i,\boldsymbol{\xi}_{k,n})\right)^T \sqrt{\frac{G}{\log(n/G)}}\left(\boldsymbol{\widehat{\theta}}_{k+1,k+G}-\boldsymbol{{\theta}}_j\right)\right\|\\
&\le \left\|\frac{1}{G}\sum_{i=k+1}^{k+G}\nabla\boldsymbol{H}(\mathbb{X}_i,\boldsymbol{\xi}_{k,n})\right\|_F\, \sqrt{\frac{G}{\log(n/G)}}\|\boldsymbol{\widehat{\theta}}_{k+1,k+G}-\boldsymbol{{\theta}}_j\|,
\end{align*}
where we used the consistency of the Frobenius matrix norm with the Euclidean vector norm in the last step.
Thus, we have shown that the left hand side diverges to infinity stochastically, while the first term on the right hand side is stochastically bounded -- both in an appropriate uniform sense. By standard arguments  this shows that 
 indeed 
\begin{align*}
	&\sqrt{\frac{G}{\log(n/G)}}
\min_{k_{j-1,n}-G\le k\le k_{j-1,n}-\varepsilon\,G
	}
	\|\boldsymbol{\widehat{\theta}}_{k+1,k+G}-\boldsymbol{{\theta}}_j\|\overset{P}{\longrightarrow}\infty.
\end{align*}
The assertion for $k_{j,n}-(1-\varepsilon)\,G\le k\le k_{j,n}$ follows analogously concluding the proof.

\end{proof}

\begin{proof}[Proof of Theorem~\ref{theoremrootnconsisalt}]
	The proof is analogous to the proof of \eqref{eq_315} where the sum 

	$	-\frac{1}{\sqrt{n}}\sum_{i=a}^b\boldsymbol{H}(\mathbb{X}_i,\tbt_{\gamma_a,\gamma_b})$
	is considered instead. The better rate compared to \eqref{eq_315} is due to the fact that a piecewise application of the central limit theorem (which follows from the mixing condition) to each regime yields
	\begin{align*}
		&\frac{1}{\sqrt{n}}\left\|\sum_{i=a}^{b}\boldsymbol{H}(\mathbb{X}_i,\tbt_{\gamma_a,\gamma_b})\right\|=O_P\left( 1 \right).
	\end{align*}
	The convergence of $\frac{1}{n}\sum_{i=a}^{b}\nabla\boldsymbol{H}(\mathbb{X}_i,\boldsymbol{\xi}_{n}^{(\gamma_a,\gamma_b)}) $ under these regularity conditions follows also e.g.\ from a piecewise application of
 Lemma~\ref{lemma_uniform} (c)(ii) (with $G$ replaced by $n$).
\end{proof}

\begin{proof}[Proof of Theorem~\ref{theorem_ass_score}]
A first order Taylor expansion yields
\begin{align*}
	&\sum_{i=k+1}^{k+G}\boldsymbol{H}(\mathbb{X}_i^{(j)},\tbt_n)-\sum_{i=k+1}^{k+G}\boldsymbol{H}(\mathbb{X}_i^{(j)},\tbt)\\
	&=\left(\sum_{i=k+1}^{k+G}\nabla\boldsymbol{H}(\mathbb{X}^{(j)}_i,\boldsymbol{\xi}_{k,n}^{(j)})\right)^T\left(\tbt_n-\tbt\right)=O_P(G/\sqrt{n})=o_P\left( \sqrt{\frac{G}{\log(n/G)}} \right),
\end{align*}
where the last line follows uniformly in $k$ by Lemma~\ref{lemma_uniform} (b). This shows the validity of
Assumption~\ref{as.replacemulti} (i). The proof of (ii) and Assumption~\ref{ass_localization_rates} (a) are analogous by using Lemma~\ref{lemma_uniform}(c) instead.

\end{proof}

\end{document}